\documentclass[journal]{IEEEtran}
\usepackage[utf8]{inputenc} 
\usepackage[T1]{fontenc}
\usepackage{url}
\usepackage{ifthen}
\usepackage{cite}
\usepackage[cmex10]{amsmath} % Use the [cmex10] option to ensure complicance
\usepackage{amssymb,mathtools}
\usepackage{graphicx}
\usepackage{multirow}
\usepackage{amsthm}
\usepackage{cuted}

\usepackage{array}
\newcolumntype{M}[1]{>{\centering\arraybackslash}m{#1}}
\newcolumntype{N}{@{}m{0pt}@{}}
\newcommand\Tstrut{\rule{0pt}{5ex}}       % "top" strut
\newcommand\Bstrut{\rule[-0.9ex]{0pt}{0pt}} % "bottom" strut
\newcommand{\TBstrut}{\Tstrut\Bstrut} % top&bottom struts

\usepackage{framed}
\usepackage[framemethod=TikZ]{mdframed}

\usepackage{enumerate}

\usepackage{tikz}
\usetikzlibrary{shapes,snakes}
\usetikzlibrary{tikzmark}
\usetikzlibrary{fit,backgrounds}
\usetikzlibrary{matrix,decorations.pathreplacing,angles,quotes}
\usetikzlibrary{patterns}
\usetikzlibrary{calc}
\usetikzlibrary{positioning}
\tikzstyle{block} = [rectangle, draw, 
    minimum width=4em, text centered, rounded corners, minimum height=1em]
\tikzstyle{rec} = [rectangle, draw]
\tikzstyle{line} = [draw, -latex]

\tikzset{snake arrow/.style=
{
decorate,
decoration={snake,amplitude=.4mm,segment length=2mm,post length=1mm}}
}

\usetikzlibrary{intersections}

\usepackage{color}

\newtheorem{lemm}{Lemma}[section]

\newtheorem{theo}{Theorem}[section]
\newtheorem{prop}{Proposition}[section]
\newtheorem{prot}{Protocol}[section]

\newtheorem{coro}{Corollary}[section]

\newtheorem{remark}{Remark}[section]

\DeclareMathOperator*{\argmin}{\arg\!\min}

\DeclareMathOperator*{\Tr}{\mathrm{Tr}}
\DeclareMathOperator*{\rank}{\mathrm{rank}}

\DeclareMathOperator{\cH}{{\mathcal{A}}}

\DeclarePairedDelimiterX\bkk[2]{\langle}{\rangle}{#1 \delimsize\vert #2}
\DeclarePairedDelimiterX\bk[2]{\langle}{\rangle}{#1 \delimsize\vert #1}
\DeclarePairedDelimiterX\kbb[2]{\vert}{\vert}{#1 \rangle\langle #2}
\DeclarePairedDelimiterX\kb[1]{\vert}{\vert}{#1 \rangle\langle #1}
\DeclarePairedDelimiter\px{\{}{\}}
\DeclarePairedDelimiter\paren{(}{)}

\DeclarePairedDelimiter\ceil{\lceil}{\rceil}

\newcommand\vN{\mathsf{n}}
\newcommand\vT{\mathsf{t}}
\newcommand\vF{\mathsf{f}}
\newcommand\vL{\mathsf{m}}
\newcommand\vM{\mathsf{m}}

\newcommand\qprot[1]{\Psi_{\mathrm{QPIR}}^{(#1)}}

\newcommand\serv{\mathtt{serv}}

\begin{document}

\title{
Capacity of Quantum Symmetric Private Information Retrieval with Collusion of All But One of Servers}
% Capacity of Quantum Private Information Retrieval with Colluding Servers}

\author{%
Seunghoan Song~\IEEEmembership{Graduate Student Member,~IEEE}
 and~Masahito~Hayashi,~\IEEEmembership{Fellow,~IEEE}% <-this % stops a space
\thanks{SS is supported by Rotary Yoneyama Memorial Master Course Scholarship (YM), Lotte Foundation Scholarship, and JSPS Grant-in-Aid for JSPS Fellows {No.~JP20J11484}. 
MH is supported in part by Guangdong Provincial Key Laboratory (Grant {No.~2019B121203002}),
a JSPS Grant-in-Aids for Scientific Research (A) No.~17H01280 and for Scientific Research (B) No.~16KT0017, and Kayamori Foundation of Information Science Advancement.}
\thanks{Seunghoan Song is with Graduate school of Mathematics, Nagoya University, Nagoya, 464-8602, Japan
(e-mail: m17021a@math.nagoya-u.ac.jp).}
\thanks{Masahito Hayashi is with 
Shenzhen Institute for Quantum Science and Engineering, Southern University of Science and Technology,
Shenzhen, 518055, China,
Guangdong Provincial Key Laboratory of Quantum Science and Engineering,
Southern University of Science and Technology, Shenzhen 518055, China,
Shenzhen Key Laboratory of Quantum Science and Engineering, Southern
University of Science and Technology, Shenzhen 518055, China,
and
the Graduate School of Mathematics, Nagoya University, Nagoya, 464-8602, Japan
(e-mail:hayashi@sustech.edu.cn).}
\thanks{This paper was presented in part at Proceedings of 
	2019 IEEE Information Theory Workshop (ITW) \cite{SH19-2}.}
}

\maketitle

\begin{abstract}
Quantum private information retrieval (QPIR) is a protocol in which a user retrieves one of multiple classical files 
	by downloading quantum systems from non-communicating $\mathsf{n}$ servers 
	each of which contains a copy of all files, 
	while the identity of the retrieved file is unknown to each server.
Symmetric QPIR (QSPIR) is QPIR in which the user only obtains the queried file but no other information of the other files.
In this paper, we consider the $(\mathsf{n}-1)$-private QSPIR 
	in which the identity of the retrieved file is secret even if any $\mathsf{n}-1$ servers collude,
and derive the QSPIR capacity for this problem which is defined as the maximum ratio of the retrieved file size to the total size of the downloaded quantum systems.
For an even number $\mathsf{n}$ of servers, we show that the capacity of the $(\mathsf{n}-1)$-private QSPIR is $2/\mathsf{n}$, when we assume that there are prior entanglements among the servers.
We construct an $(\mathsf{n}-1)$-private QSPIR protocol of rate $\ceil{\mathsf{n}/2}^{-1}$ and 
	prove that the capacity is upper bounded by $2/\mathsf{n}$ even if any error probability is allowed.
%The proposed protocol is symmetric QPIR protocol, i.e., the user obtains no information of non-queried files.
The $(\mathsf{n}-1)$-private QSPIR capacity is strictly greater than the classical counterpart.
\end{abstract}

\section{Introduction}

A private information retrieval (PIR) protocol is a protocol in which a user retrieves a file from servers without revealing the identity of the retrieved file.
Since it was first proposed by the paper \cite{CGKS98}, it has been studied in classical settings \cite{CMS99,Lipmaa10,BS03,DGH12,CHY15} and quantum settings \cite{KdW03,KdW04, Ole11, BB15, LeG12}.
Especially, in the last few years, the classical PIR capacity has been extensively studied, which is the maximum rate of the retrieved file size over the download size when the file size is arbitrarily large and the upload size, i.e., the total size of queries, is negligible to the download size.
% Starting from 
The paper \cite{SJ17} derived the PIR capacity for the most trivial setting in which each server contains the replicated file set.
Several PIR capacities have been derived, e.g., %(see Table \ref{tab:compare}), e.g., 
	when some servers may collude \cite{SJ18}, 
	when the user is also 
		prohibited to obtain any information about the non-queried files \cite{SJ17-2} ({\em symmetric} PIR, SPIR),
	and when the file set is coded and distributed to the servers \cite{BU18}. 
% When the server is 
% the capacity when the user is also prohibited to obtain the non-targeted files \cite{SJ17-2} ({\em symmetric} PIR),
Moreover, 
many other papers \cite{FHGHK17, KLRG17, LKRG18, TSC18, TSC18-2, WS17, WS17-2,BU19, Tandon17} have studied the PIR capacity and the capacity-achieving protocols.
% many other papers \cite{FHGHK17, KLRG17, LKRG18, TSC18} studied the CPIR capacity and the capacity-achieving protocols.
% many other papers \cite{FHGHK17, KLRG17, LKRG18, TSC18, TSC18-2, WS17, WS17-2, BU19, Tandon17} studied the CPIR capacity and the capacity-achieving protocols.
% the papers \cite{FHGHK17, KLRG17, LKRG18} constructed protocols achieving the CPIR capacity with coded and distributed servers.
% and the paper \cite{TSC18} optimized the upload cost of the capacity-achieving CPIR for the trivial setting.

As a quantum extension of the PIR capacity in \cite{SJ17} and the SPIR capacity in \cite{SJ17-2},
	the paper \cite{SH19} proved that the quantum PIR (QPIR) capacity and the QSPIR capacity are both $1$ for the replicated servers.
% The capacity $1$ holds whether it is 
The above QPIR capacity is strictly greater than the classical counterparts \cite{SJ17} and \cite{SJ17-2}. %, it is advantageous to introduce quantum methods to PIR.
However, %since the QPIR capacity in \cite{SH19} only treated the case for replicated servers without collusion,
it still needs to be clarified whether QPIRs in other settings have advantages over the classical counterparts in the sense of the PIR capacities.
For instance, in the PIR for $\vN$ replicated servers where each server contains $\vF$ files and any $\vT$ servers may collude ({\em $\vT$-private PIR}),
% \footnote[1]{It is originally called {\em $\vT$-private private information retrieval} in \cite{SJ18} but we use the term {\em $\vT$-colluded private information retrieval} for clarity.}
	the PIR capacity is $\paren*{1-\vT/\vN}/\big(1-\paren*{\vT/\vN}^{\vF}\big)$ \cite{SJ18}
	and the $\vT$-private SPIR capacity is $(\vN-\vT)/\vN$,
	but the QPIR capacity is unknown.

\begin{table}[t]   \label{tab:compare}
\centering
\caption[caption]{Classical and quantum PIR capacities}
\begin{tabular}{|M{2cm}|M{2.8cm}|M{2.3cm}|N}
\cline{1-3}
                                         &   Classical PIR  &  Quantum PIR  &	\TBstrut
                \\[0.3em]
\cline{1-3}
PIR    &  $\displaystyle \frac{1\!-\!1/\vN}{1\!-\!\paren*{1/\vN}^{\vF}}$     \cite{SJ17} 
                      &  $\displaystyle1$ \cite{SH19} & \TBstrut
                \\
                \cline{1-3}
SPIR    &  $\displaystyle \frac{\vN-1}{\vN}$     \cite{SJ17} 
                      &  $\displaystyle1$ \cite{SH19} & \TBstrut
                \\
\cline{1-3}
$\vT$-private SPIR
        &   $\displaystyle \frac{\vN-\vT}{\vN} $ \cite{WS17-2}
        &   $\displaystyle \frac{2}{\vN}$ $\ddagger$  & \TBstrut
%         for even $\vN$ $\vT=\vN-1$ 
                \\
\cline{1-3}
% Coded servers by $(\vN,\vkc)$ MDS code  
%         &  $\displaystyle\frac{1-\vkc/\vN}{1-\paren*{\vkc/\vN}^{\vF}} $ \cite{BU18}
%         &  $-$   &
%                 \\[0.0em]      
% \cline{1-3}
% \multicolumn{4}{l}{\scriptsize $\ast$ $\vN$: num. of servers, $\vF$: num. of files, $\vT$: num. of colluding servers.}\\
% &&&\\[-1em]
\multicolumn{3}{l}{}\\[-0.9em]
\multicolumn{3}{l}{\scriptsize $\ast$ $\vN$, $\vF$, and $\vT$: the number of servers, files, and colluding servers, respectively.}\\
\multicolumn{3}{l}{\scriptsize $\ddagger$ This paper derives the capacity for any even number $\vN$ and the number $\vT=$}\\
\multicolumn{3}{l}{\enskip \scriptsize $\vN-1$.}
\end{tabular}
\end{table}

In this paper, we prove the $(\vN-1)$-private QSPIR capacity is $2/\vN$ 
% for $\vN$ replicated servers where any $\vN-1$ servers may collude, 
for any even number $\vN$.
% In both of protocol construction and converse proof,
% 	we consider the server secrecy that the user obtains no information of the non-targeted files.
For any number of servers $\vN\geq2$, we construct an $(\vN-1)$-private QSPIR protocol with the rate $\ceil{\vN/2}^{-1}$, no error, and the perfect secrecy.
We also prove the strong converse bound in which 
	the $(\vN-1)$-private QSPIR capacity is upper bounded by $2/\vN$
even if we allow any asymptotic error probability less than $1$.
Since the $(\vN-1)$-private PIR capacity for the infinite number of files and the $(\vN-1)$-private SPIR capacity are $1/\vN$, 
	our QSPIR capacity implies the quantum advantage in PIR with colluding servers.

Our capacity-achieving protocol is a generalization of the QSPIR protocol in \cite{SH19}. 
The protocol in \cite{SH19} extended the classical PIR protocol \cite{CGKS98} by the idea of the superdense coding \cite{BW92}.
Similarly, our protocol extends a classical $(\vN-1)$-private PIR protocol explained below by the idea of
% the protocol \cite{SH19} and 
	the superdense coding \cite{BW92} and 
	the quantum teleportation \cite{BBCJPW93}.
The classical PIR protocol we extend is described as follows.
Let $(\log \vL)$-bit files $W_1,\ldots,W_{\vF}$ be contained in each of $\vN$ servers and the queries $Q_1,\ldots,Q_{\vN-1}$ be independently and uniformly chosen subsets of $\{1,\ldots,\vF\}$.
To retrieve the $K$-th file, the user chooses $Q_{\vN}$ which satisfies $\bigoplus_{t=1}^{\vN} Q_{t} = \{K\}$, where $\bigoplus$ is the symmetric difference,
	and sends the queries $Q_1,\ldots,Q_{\vN}$ to each server.
For each $t\in\{1,\ldots,\vN\}$, 
	the $t$-th server returns $H_t \coloneqq \sum_{i\in Q_t} W_i$ to the user
	and then the user can retrieve $W_K = \sum_{t=1}^{\vN} H_t$,
	where both summations are with respect to the addition modulo $2$.
%	where the addition is modulo $2$ in both summations.
The protocol is private because the collection of any $\vN-1$ variables in $Q_1,\ldots,Q_{\vN}$ is independent of the query index $K$.

Our capacity-achieving protocol has several remarkable properties. % over its classical counterpart in \cite{SJ18}.
First, 
% symmetric QPIR protocol since the user obtains no information of non-targeted files.}
our protocol is a symmetric QPIR protocol, i.e., it guarantees the server secrecy in which no file information other than the queried one is transmitted to the user.
Second, 
the upload cost of our protocol is $\vN\vF$ bits, which is linear for the number of servers $\vN$ and the number of files $\vF$ but independent of the file size $\vL$. %, whereas that of \cite{SJ18} \myred{is }.
%whereas the protocol in \cite{SJ18} obtains \myred{some information of non-retrieved files.}
Third, our protocol requires the file size $\vL=2^{2\ell}$, i.e., $2\ell$ bits, for any positive integer $\ell$,
whereas the $(\vN-1)$-private PIR protocol in \cite{SJ18} requires the file size $\vL=q^{\vN^{\vF}}$ depending on $\vN$ and $\vF$ for a sufficiently large prime power $q$. %, which depends on $\vN$ and $\vF$.

Following the conference version of this work, 
	the paper \cite{AHPH20} proposed a QSPIR protocol for coded and colluding servers 
		which works for any $[\vN,k]$-MDS code and secure against $\vT$-collusion with $\vT=\vN-k$.
Their protocol is an extension of the protocol of this paper by combination with the classical PIR protocol \cite{FHGHK17},
	and it achieves better rates than the classical counterparts \cite{FHGHK17,WS17-2}.
The paper \cite{SH20} improved the capacity result of this paper for any number of colluding servers: For any $1< \vT< \vN$, the $\vT$-private QSPIR capacity is $\min\{1,2(\vN-\vT)/\vN\}$.
Even though the capacity of \cite{SH20} generalizes that of this paper,
	our protocol and converse proof have the following advantages compared to \cite{SH20}.
First, whereas the protocol of \cite{SH20} requires multipartite entanglement as prior entanglement,
	our protocol only requires multiple copies of bipartite entangled states. % instead of a large entangled state.
	Since the bipartite entanglement is more reliably generated with current technology,
	our construction is more suitable for the implementation on near-term quantum devices than that of \cite{SH20}.
Second, our protocol is more constructive and easier to understand than the protocol of \cite{SH20}.
	Our protocol is a combination of two simple protocols: quantum teleportation \cite{BBCJPW93} and superdense coding \cite{BW92}, 
		which have been experimentally realized in \cite{QT1,QT2} and \cite{Dense1,Dense2}, respectively.
	On the other hand, the protocol in \cite{SH20} is constructed with the sophisticated method of stabilizer formalism.
	Thus, our protocol is more accessible to the experimentalists and the theorists who are not familiar with the stabilizer formalism.
Lastly, the converse proof of our paper is much simpler than that of \cite{SH20}
	because we only prove for the case of $\vT=\vN-1$.

% However, 
% 	the protocol in this paper has advantage over that of \cite{SH20} 
% 		since our protocol requires many copies of bipartite entangled state,
% 	whereas the protocol of \cite{SH20} requires multipartite entanglement among all servers as prior entanglement, which is hard 
% % 	whereas the protocol in this paper requires many copies of bipartite entangled state.

The rest of the paper is organized as follows.
Section \ref{sec:problem} defines the QPIR protocol and presents our main theorems for the $(\vN-1)$-private QSPIR capacity.
Section \ref{sec:prelim} is preliminaries for the protocol construction and
Section \ref{sec:protocol} constructs the QSPIR protocol with $\vN-1$ colluding servers.
Section \ref{sec:converse} proves the strong converse bound when the perfect secrecy is guaranteed.

% In classical PIR, 
% the paper \cite{} showed that 
% the capacity is 
% $\paren*{1-\frac{\vT}{\vN}}/\big(1-\paren*{\frac{\vT}{\vN}}^{\vF}\big)$
% % $\paren*{ + \paren*{\frac{\vT}{\vN}}^{2} + \cdots + \paren*{\frac{\vT}{\vN}}^{\vF-1} }^{-1}$
% % $\paren*{1+\frac{\vT}{\vN} + \paren*{\frac{\vT}{\vN}}^{2} + \cdots + \paren*{\frac{\vT}{\vN}}^{\vF-1} }^{-1}$
% when there is $\vN$ replicated servers 
% each containing $\vF$ files where any $\vT$ servers may collude.

% PIR capacity with colluding servers

% 
% for $\vN$ servers each containing where any $\vN-1$ servers.

\subsubsection*{Terms and Notations}
\textit{
We summarized the basic notions of quantum information theory in Appendix~A.
For two sets $A$ and $B$, define $A\oplus B \coloneqq (A\backslash B) \cup (B\backslash A)$.
The identity operator on any space $\cH$ is denoted by $\mathsf{I}_{\cH}$, or simply by $\mathsf{I}$ if there is no confusion.
% For pure states on composite system, $|i,i\rangle \coloneqq | i\rangle \otimes |i \rangle$.
For a state $\rho$ on a composite system $\cH\otimes\mathcal{B}$, %on the composite system of $\mathcal{A}$ and $\mathcal{B}$,
denote the quantum mutual information between the systems $\mathcal{A}$ and $\mathcal{B}$ by $I(\mathcal{A};\mathcal{B})_{\rho}$.
The classical mutual information is denoted as $I(X;Y)$ without subscript.
For states $\rho_r$ on a system $\cH$ which depend on the value $r$ of a random variable $R$,
	we denote $I(R;\mathcal{A})_{\rho_R} \coloneqq I(R;\mathcal{A})_{\sum_{r} p_r |r\rangle\langle r|\otimes \rho_r}$.
For a matrix $A$, we denote the transpose of $A$ by $A^{\top}$ and the conjugate transpose of $A$ by $A^{\dagger}$.
}

\section{Problem Statement and Main Results}    \label{sec:problem}

In this section, we review the QPIR protocol given in the paper \cite{SH19}, which is defined in the same way for the $(\vN-1)$-private QPIR except for the security measures.
% After defining the security parameters of colluded QPIR protocols, 
Then, we present two main theorems for the $(\vN-1)$-private QSPIR capacity.

\subsection{Description of QPIR Protocol} %(Fig. \ref{fig:protocol_model})}

% We review QPIR protocol defined in \cite{SH19}.

\begin{figure}[t]
\begin{center}
        \resizebox {0.9\linewidth} {!} {
\includegraphics[width=0.7\linewidth]{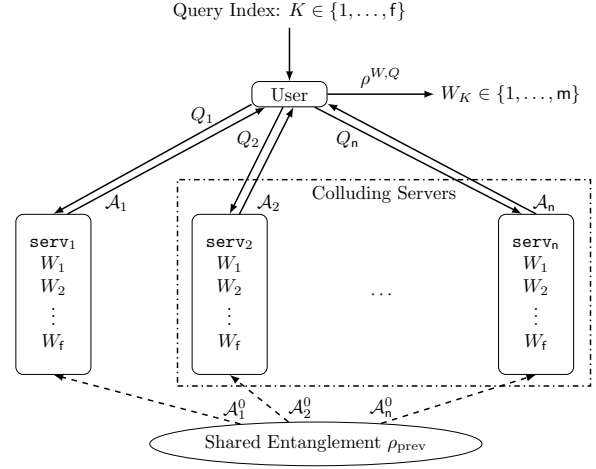} 
	}
\caption{Quantum private information retrieval protocol with collusion of all but one of servers.
The user does not know which $\vN-1$ servers collude.
% The set of colluding $\vN-1$ servers (dash-dotted) is unknown arbitrary.
}\label{fig:protocol_model}
\end{center}
\end{figure}

Let $W_1,\ldots, W_{\vF}$ be uniformly and independently distributed in $\{0,\ldots,\vL-1\}$.
Each of non-communicating $\vN$ servers $\mathtt{serv}_1,\ldots,\mathtt{serv}_{\vN}$ contains the replicated $\vF$-file set $W\coloneqq(W_1,\ldots, W_{\vF})$. 
In addition, each server $\mathtt{serv}_t$ contains a quantum system $\cH_t^{0}$
and the $\vN$ servers share an arbitrary entangled state $\rho_{\mathrm{prev}}$ on $\cH_1^{0}\otimes \cdots\otimes\cH_{\vN}^{0}$.
A user chooses the query index $K\in\{1,\ldots,\vF\}$ uniformly at random.
The aim of the QPIR protocol is for the user to retrieve the file $W_K$ from the servers.

For this purpose, the user and the servers perform the following process.
First, the user chooses a random variable $R_{\mathrm{user}}$ from a set $\mathcal{R}_{\mathrm{user}}$ and encodes the queries $Q_1,\ldots,Q_\vN$ by the user encoder $\mathsf{Enc}_{\mathrm{user}}$ as follows:
\begin{align*}
\mathsf{Enc}_{\mathrm{user}} (K,R_{\mathrm{user}}) = Q =  (Q_1,\ldots,Q_{\vN})  \in \mathcal{Q}_1\times\cdots\times \mathcal{Q}_{\vN},
\end{align*}
where $\mathcal{Q}_t$ for any $t\in\{1,\ldots,\vN\}$ is a finite set for describing possible query indexes to the server $\mathtt{serv}_t$.
Then, for any $t\in\{1,\ldots,\vN\}$, each query $Q_t$ is sent to the server $\mathtt{serv}_t$.
%\myred{Let $\cH_1,\ldots,\cH_{\vN}$ be the same-dimensional systems.}
After receiving the query $Q_t$, each server $\mathtt{serv}_t$ prepares a 
trace-preserving and completely positive (TP-CP) map\footnote[1]{
Server operations induced by $\mathsf{Enc}_{\mathrm{serv}}$ may contain
random operations and measurements because TP-CP maps contain random operations and measurements.}
$\Lambda_t$ from $\cH_t^{0}$ to $\cH_t$ by the server encoder $\mathsf{Enc}_{\mathrm{serv}_t}$, i.e., 
\begin{align*}
\mathsf{Enc}_{\mathrm{serv}_t} (Q_t, W) =  \Lambda_t
\end{align*}
and the resultant state on $\cH\coloneqq\cH_1\otimes \ldots\otimes \cH_{\vN}$ is 
\begin{align*}
\rho^{W,Q} \coloneqq \Lambda_1\otimes\cdots\otimes \Lambda_{\vN} (\rho_{\mathrm{prev}})  %\in \mathcal{S}(\cH_1\otimes \cdots \otimes \cH_{\vN}).
\end{align*}
Then, each server $\serv_t$ sends $\cH_t$ to the user.
%finally
Depending on $K$ and $Q$, the user chooses the decoder $\mathsf{Dec}(K,Q)$, which is a 
	positive-operator valued measure (POVM) $\{\mathsf{Y}_i \mid i\in \{0,\ldots, \vL\}\}$ on $\cH$.
By performing the POVM measurement $\mathsf{Dec}(K,Q)$ on $\cH$,
	the user obtains the measurement outcome $\hat{W}_K\in\{0,\ldots,\vL\}$ as the retrieval result, where the outcome $\vL$ is considered as decoding failure.
With probability $\Tr \rho^{W,Q}\mathsf{Y}_{W_K}$, the queried file is obtained, i.e., $\hat{W}_K = W_K$.

%retrieves the queried file $W_K$ by the decoder $\mathsf{Dec}(K,Q)$ defined as follows:
%\myred{$\mathsf{Dec}(K,Q)$ is 
%	%\footnote[2]{a POVM is a set $\{\mathsf{Y}_i \mid i\in \{0,\ldots, \vL-1\}\}$ of positive semidefinite matrices on $\cH$ satisfying $\sum_{i=0}^{\vL-1} \mathsf{Y}_i \leq \mathsf{I}_{\cH}$.}
%%	%, which is the set of Hermitian matrices $\{M_{x}\}_{x\in\mathcal{X}}$ on $\cH$ such that 
%%	%\begin{align}
%%	%M_x \geq 0, \quad \sum_{x \in \mathcal{X}}  M_x = I_{\cH}.
%%	%\end{align}
%	depending on the query index $K$ and the query $Q$.

Given the number of servers $\vN$ and the number of files $\vF$,
% Given the numbers $\vN$ and $\vF$ of servers and files,
a QPIR protocol is described by the four-tuple 
$$\Psi_{\mathrm{QPIR}}^{(\vL)} \coloneqq (\rho_{\mathrm{prev}}, \mathsf{Enc}_{\mathrm{user}}, \mathsf{Enc}_{\mathrm{serv}}, \mathsf{Dec}),$$
where $\mathsf{Enc}_{\mathrm{serv}} \coloneqq (\mathsf{Enc}_{\mathrm{serv}_{1}}, \ldots,\mathsf{Enc}_{\mathrm{serv}_{\vN}})$.
Note that the QPIR protocol $\Psi_{\mathrm{QPIR}}^{(\vL)}$ characterizes the process in the previous paragraph.
The upload cost, download cost, and rate of a protocol $\Psi_{\mathrm{QPIR}}^{(\vL)}$ are defined by
\begin{gather*}
U(\qprot{\vL}) = \prod_{t=1}^{\vN} |\mathcal{Q}_t| , \quad
D(\qprot{\vL}) = \dim \bigotimes_{t=1}^{\vN} \cH_t, \\
R(\qprot{\vL}) = \frac{\log\vL}{\log D(\qprot{\vL})} .
\end{gather*}

\subsection{Security Measures and $(\vN-1)$-Private QPIR Capacity}
% In the remainder of the paper, we assume that 
% the number $\vN$ of servers and the number $\vF$ of files are fixed, and
We define security measures and the capacity of the $(\vN-1)$-private QPIR. % when any $\vN-1$ servers may collude to reveal the identity $K$ of the retrieved file $W_K$.

\subsubsection{Security Measures}
$(\vN-1)$-Private QPIR is QPIR in which any $\vN-1$ servers may collude to determine $K$ but it should not be leaked even if the user does not know which servers are colluding.
Furthermore, we also consider the server secrecy in which the user only obtains the queried file $W_K$ but no information of other files.
Thus, we evaluate the security of a protocol $\Psi_{\mathrm{QPIR}}^{(\vL)}$ by the error measure, the server secrecy measure, and the user secrecy measure.
The error measure is defined by %the smallest real number such that
\begin{align*}
\alpha(\qprot{\vL}) \coloneqq P_{\mathrm{err}}(\Psi_{\mathrm{QPIR}}^{(\vL)}) = \mathbb{E}_{W,K,Q} (1-\Tr \rho^{W,Q}\mathsf{Y}_{W_K}),
\end{align*}
where $P_{\mathrm{err}}(\Psi_{\mathrm{QPIR}}^{(\vL)} )$ is the average error probability of the protocol.
The server and user secrecy measures $\beta(\qprot{\vL})$ and $\gamma(\qprot{\vL})$ are defined 
respectively as
\begin{align}
\beta(\qprot{\vL})  &\coloneqq  \max_{  i,k\in\{1,\ldots, \vF\}: i\neq k  } I(W_i;\cH|K=k)_{\rho^{W,Q}} ,    \label{eq:serv_sec}\\
\gamma(\qprot{\vL}) &\coloneqq  \max_{t\in\{1,\ldots,\vN\}} I(K;Q_t') , \label{eq:user_sec} 
%d(P_{Q',K=k},P_{Q',K=i}) & \leq \gamma \quad \forall i\neq k \in\{1,\ldots, \vF\} , 
% \label{eq:user_sec} 
\end{align}
where $Q_t'$ is the collection of queries to all servers other than $\mathtt{serv}_t$.
If these measures are zero, the protocol $\Psi_{\mathrm{QPIR}}^{(\vL)}$ is called a $(\vN-1)$-private QSPIR protocol,
	and if $\alpha(\qprot{\vL}) =  \gamma(\qprot{\vL}) = 0$, it is called a $(\vN-1)$-private QPIR protocol.

\subsubsection{$(\vN-1)$-private QPIR Capacity}

The $(\vN-1)$-private QPIR capacity is defined with the security and upload constraints.
For any $\alpha$, $\beta$, $\gamma$, $\theta \geq 0$, the {\em asymptotic} and {\em exact security-constrained $(\vN-1)$-private QPIR capacities} are defined by

\begin{align*}%\\[-3em]
C_{\mathrm{asymp}}^{\alpha,\beta,\gamma,\theta} 
				& \coloneqq \sup_{\eqref{con1}} 
				\liminf_{\ell\to\infty} R(\Psi_{\mathrm{QPIR}}^{(\vM_\ell)}),\\
C_{\mathrm{exact}}^{\alpha,\beta,\gamma,\theta} 
				& \coloneqq \sup_{\eqref{con2}}
				\liminf_{\ell\to\infty} R(\Psi_{\mathrm{QPIR}}^{(\vM_\ell)}),
% \\[-2em]
\end{align*}
% \end{strip}
% \!\!\!where the supremum is taken for sequences $\{\vM_\ell\}_{\ell=1}^{\infty}$ such that $\lim_{\ell\to\infty} \vM_\ell = \infty$ 
where the supremum is taken for sequences $\{\vL_\ell\}_{\ell=1}^{\infty}$ such that $\lim_{\ell\to\infty} \vL_\ell = \infty$ and sequences $\{\Psi_{\mathrm{QPIR}}^{(\vL_\ell)}\}_{\ell=1}^{\infty}$ of QPIR protocols
	to satisfy either \eqref{con1} or \eqref{con2} given by
\begin{align}    \label{con1} %\tag{$\ast$} 
\begin{split} 
\!\!\!\limsup_{\ell\to\infty}\alpha(\qprot{\vL_{\ell}})\leq \alpha,,  \enskip
& \limsup_{\ell\to\infty}\beta(\qprot{\vL_{\ell}})\leq \beta,\\
\!\!\!\limsup_{\ell\to\infty}\gamma(\qprot{\vL_{\ell}})\leq \gamma,, \enskip
& \limsup_{\ell\to\infty} \frac{\log U(\Psi_{\mathrm{QPIR}}^{(\vL_\ell)})}{\log D(\Psi_{\mathrm{QPIR}}^{(\vL_\ell)})} \leq  \theta,
\end{split}
\end{align}
and 
\begin{gather}  \label{con2} 
\begin{split}
 \alpha(\qprot{\vL_{\ell}})\leq \alpha, \enskip
& \beta(\qprot{\vL_{\ell}})\leq \beta,   \\
 \gamma(\qprot{\vL_{\ell}})\leq \gamma, \enskip
& \limsup_{\ell\to\infty} \frac{\log U(\Psi_{\mathrm{QPIR}}^{(\vL_\ell)})}{\log D(\Psi_{\mathrm{QPIR}}^{(\vL_\ell)})} \leq  \theta
.
\end{split}
\end{gather}
The parameters $\alpha,\beta,\gamma,\theta$ are the upper bounds of the error, server secrecy, user secrecy, and upload cost, respectively,
	and the two capacities 
	$C_{\mathrm{asymp}}^{\alpha,\beta,\gamma,\theta}$,
	$C_{\mathrm{exact}}^{\alpha,\beta,\gamma,\theta}$ 
	are defined as the supremum of QPIR rates for all QPIR protocols satisfying the upper bounds asymptotically and exactly, respectively.
Since any protocols satisfying the upper bounds $\alpha,\beta,\gamma,\theta$ exactly also satisfy the bounds asymptotically,
	we have the inequality
	$C_{\mathrm{asymp}}^{\alpha,\beta,\gamma,\theta} \geq C_{\mathrm{exact}}^{\alpha,\beta,\gamma,\theta} \geq C_{\mathrm{exact}}^{0,0,0,0}$.

%	are defined as the supremum of QPIR rates for all QPIR protocols satisfying the upper bounds.
%The capacity $C_{\mathrm{asymp}}^{\alpha,\beta,\gamma,\theta}$ is the capacity for the case where the upper bounds are asymptotically satisfied and $C_{\mathrm{exact}}^{\alpha,\beta,\gamma,\theta}$ is the capacity for the case where the upper bounds are exactly satisfied.

\subsection{Main Results}
Two main theorems of the paper are as follows.
\begin{theo} \label{theo:main2}
For any $\vN\geq 2$ servers and $\vF\geq 2$ files, 
	%where the collection of any $\vN-1$ servers may collude,
	there exists a $(\vN-1)$-private QSPIR protocol with 
the rate $\ceil*{\vN/2}^{-1}$,
zero security measures, 
$\vN\vF$-bit upload cost, 
and 
$2\ell$-bit files for any integer $\ell\geq 1$.
\end{theo}
Section \ref{sec:protocol} constructs the protocol that achieves the performance given in Theorem \ref{theo:main2}. 
When $\vN=2$, the protocol in Section \ref{sec:protocol} corresponds to the protocol in \cite{SH19}.

\begin{theo}[Capacity of $(\vN-1)$-private QPIR] \label{theo:main}
For any $\vN\geq 2$ servers and $\vF\geq 2$ files,
% For the $(\vN-1)$-colluded QPIR with $\vN$ servers,
% For any even positive integer $\vN$,
the $(\vN-1)$-private QPIR capacity satisfies %for $\vN$ servers of which any $\vN-1$ servers may collude is 
\begin{align}
C_{\mathrm{asymp}}^{\alpha,\beta,\gamma,\theta} \ge
C_{\mathrm{exact}}^{\alpha,\beta,\gamma,\theta} \ge
C_{\mathrm{exact}}^{0,0,0,0} 
& \ge \ceil*{\frac{\vN}{2}}^{-1}, \label{H1}\\
C_{\mathrm{asymp}}^{\alpha,0,0,\theta} 
& \le \frac{2}{\vN} \label{H2}
\end{align}
for any $\alpha\in [0, 1)$ and $\beta, \gamma, \theta\geq 0$.
\end{theo}

The last inequality in \eqref{H1} follows from Theorem \ref{theo:main2} 
and
the inequality \eqref{H2} will be proved in Section \ref{sec:converse}.
From Theorem~\ref{theo:main}, we obtain the following corollary. %, which is the main result of this paper.
\begin{coro}
For any even number of servers $\vN$ and any number of files $\vF\geq 2$,
	the $(\vN-1)$-private QSPIR capacity is $2/\vN$.%for $\vN$ servers of which any $\vN-1$ servers may collude is 
\end{coro}

\section{Preliminaries for Protocol Construction}      \label{sec:prelim}

In this section, we prepare two simple protocols to describe our $(\vN-1)$-private QSPIR protocol.

\subsection{Preliminaries on States, Operations, and Measurements} \label{subsec:prelim}

A qubit $\cH$ is a two-dimensional complex Hilbert space spanned by an orthonormal basis $\{|0\rangle, |1\rangle\}$.
Define the maximally entangled state $|\Phi\rangle$ on two qubits $\cH\otimes \cH$ by
\begin{align*}
|\Phi \rangle \coloneqq  \frac{1}{\sqrt{2}} \sum_{i=0}^{1}| i\rangle\otimes |i \rangle.
\end{align*}
For any $a,b\in\mathbb{Z}_2$, define Pauli operations on $\cH\otimes\cH$ by
\begin{gather*}
\mathsf{X} \coloneqq \sum_{i=0}^{1} |i+1\rangle \langle i |,\
\mathsf{Z} \coloneqq \sum_{i=0}^{1} (-1)^{i} |i\rangle \langle i |,\
\mathsf{W}(a,b) \coloneqq \mathsf{X}^a\mathsf{Z}^b.
\end{gather*}
It can be easily checked that these operations satisfy the relations
\begin{align}
\mathsf{Z}^b\mathsf{X}^a &= (-1)^{ab} \mathsf{X}^a\mathsf{Z}^b = (-1)^{ab}\mathsf{W}(a,b),\\
\mathsf{W}(a,b)^\dagger & %= \mathsf{Z}^{-b}\mathsf{X}^{-a} 
= \mathsf{W}(a,b)^\top = (-1)^{ab} \mathsf{W}(a,b), \label{property0}\\
\mathsf{W}(a_1,b_1)\mathsf{W}(a_2,b_2) &= (-1)^{b_1a_2} \mathsf{W}(a_1+a_2, b_1+b_2) \label{eq:sum}.
\end{align}
Applying \eqref{eq:sum} twice, we also have
\begin{align}
&\mathsf{W}(a_1,b_1)\mathsf{W}(a_2,b_2) \nonumber \\
&= (-1)^{b_1a_2+a_1b_2} \mathsf{W}(a_2,b_2)\mathsf{W}(a_1,b_1)
. \label{eq:comm2}
\end{align}

For any matrix $\mathsf{T} = \sum_{i,j=0}^{1} t_{ij} |i\rangle\langle j|$ on $\cH$,
we define the vector $|\mathsf{T}\rangle$ in $\cH\otimes \cH$ by
\begin{align}
|\mathsf{T}\rangle \coloneqq \sum_{i,j=0}^{1} t_{ij} |i\rangle \otimes |j\rangle.
\end{align}
With this notation, the maximally entangled state is written as $|\Phi\rangle = (1/\sqrt{2})|\mathsf{I}\rangle$.
Since $\mathsf{T}^{\top} = \sum_{i,j=0}^{1} t_{ij} |j\rangle\langle i|$,
% where $\mathsf{I}$ is the identity matrix on $\cH$.
it holds $|\mathsf{T}\rangle = (\mathsf{T}\otimes \mathsf{I}) |\mathsf{I}\rangle = (\mathsf{I}\otimes \mathsf{T}^{\top}) |\mathsf{I}\rangle$.
Moreover, for any unitaries $\mathsf{U},\mathsf{V}$ on $\cH$,
\begin{align}
(\mathsf{U}\otimes \mathsf{V}) |\mathsf{T}\rangle &= 
% (\mathsf{U}\mathsf{T}\otimes \mathsf{V})|\mathsf{I}\rangle = (\mathsf{U}\mathsf{T}\mathsf{V}^{\top}\otimes \mathsf{I})|\mathsf{I}\rangle = 
|\mathsf{U}\mathsf{T}\mathsf{V}^{\top}\rangle,  \\
(\mathsf{U}\otimes \overline{\mathsf{U}}) |\mathsf{I}\rangle &= |\mathsf{U}\mathsf{U}^\dagger\rangle = |\mathsf{I}\rangle. \label{eq:uuI}
\end{align}
For the maximally entangled state $|\Phi\rangle$ on $\cH\otimes \cH$, the Pauli operation $\mathsf{W}(a,b)$ on the first (second) qubit can be translated to the operation $\mathsf{W}(a,b)$ on the second (first) qubit because
\begin{align}
 \lefteqn{(\mathsf{I}\otimes \mathsf{W}(a,b))|\Phi\rangle = (\mathsf{W}(a,b)^{\top} \otimes \mathsf{I})|\Phi\rangle} \nonumber  \\
&= (\mathsf{W}(a,b)^{\dagger} \otimes \mathsf{I})|\Phi\rangle
= (-1)^{ab}(\mathsf{W}(a,b) \otimes \mathsf{I})|\Phi\rangle.  \label{eq:exchange}
% (\mathsf{I}\otimes \mathsf{W}(a,b))|\Phi\rangle &= (-1)^{ab} 
\end{align}

The following proposition is a case of \cite[Proposition III.1]{SH19} for qubits.
\begin{prop}   \label{prop:pvm}
The set 
$$\px*{ (\mathsf{I}\otimes \mathsf{W}(a,b))|\Phi \rangle
        =   \frac{1}{\sqrt{2}}\sum_{r=0}^{1} (-1)^{rb} | r, r+a \rangle
        \ \Bigg|\ a,b\in\mathbb{Z}_2 }$$
is an orthonormal basis of $\cH\otimes \cH$.
\end{prop}
From Proposition \ref{prop:pvm}, we can define the projection-valued measure (PVM)
\begin{align}
\mathbf{M}_{\mathbb{Z}_2^2} &\coloneqq \{  M_{(a,b)} 
         \ | \  a,b\in\mathbb{Z}_2 \} \label{M2}
         ,
\end{align}
where $M_{(a,b)} \coloneqq (\mathsf{I}\otimes \mathsf{W}(a,b)) |\Phi \rangle
	\langle \Phi | (\mathsf{I}\otimes \mathsf{W}(a,b)^{\dagger})       $.  
% which is used in the following protocols.

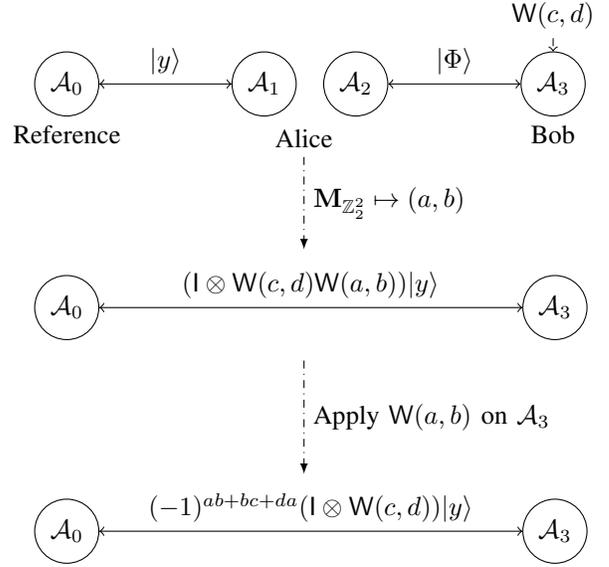
\begin{figure}[t]
\begin{center}
\begin{tikzpicture}[node distance = 3.3em, every text node part/.style={align=center}, auto, scale=0.7]
    % Place nodes
    \node [draw, circle] (channel) {$\cH_0$};
    \node [below=0em of channel] (system) {Reference};
    \node [draw, circle, right=5em of channel] (cha) {$\cH_1$};
    \node [below right=0.5em and -0.8em of cha] (alice) {Alice};
    \node [draw, circle, right=1em of cha] (b1) {$\cH_2$};
    \node [draw, circle, right=5em of b1] (b2) {$\cH_3$};
    \node [below=0em of b2] (system) {Bob};
    
    % Draw edges
    \path [line, <->] (channel) --node {$|y\rangle$} (cha);
    \path [line, <->] (b1) --node {$|\Phi\rangle$} (b2);
%     \path [line] (decoder) --node {$\mathcal{D}(\sigma)$} (5,0);

    \node [above=0.5em of b2] (op3) {$\mathsf{W}(c,d)$};
    \path [line, ->,dashed] (op3) -- (b2);

    \node [below=3.5em of alice] (measure) {};
    \path [line, =>,dashdotted] (alice) --node {$\mathbf{M}_{\mathbb{Z}_2^2}\mapsto (a,b)$} (measure);

    \node [draw, circle, below=6em of channel] (channel2) {$\cH_0$};
%     \node [below=0em of channel2] (system2) {$\mathtt{serv}_1$//($H_3$)};
%     \node [right=5em of channel2] (cha2) {};
%     \node [below right=0em and -0.8em of cha2] (system2) {$\mathtt{serv}_2$//($H_1$)};
%     \node [right=1em of cha2] (b12) {};
    \node [draw, circle, right=16em of channel2] (b22) {$\cH_3$};
%     \node [below=0em of b2] (system2) {$\mathtt{serv}_3$//($H_2$)};
    
    % Draw edges
    \path [line, <->] (channel2) --node {$(\mathsf{I}\otimes \mathsf{W}(c,d)\mathsf{W}(a,b))|y\rangle$} (b22);
%     \path [line, <->] (b12) --node {$|\Phi\rangle$} (b22);

    \node [below=7em of alice] (x) {};
    \node [below=12em of alice] (step2) {};
    \path [line, =>,dashdotted] (x) --node {Apply $\mathsf{W}(a,b)$ on $\cH_3$} (step2);

    \node [draw, circle, below=6em of channel2] (channel3) {$\cH_0$};
%     \node [below=0em of channel2] (system2) {$\mathtt{serv}_1$//($H_3$)};
%     \node [right=5em of channel2] (cha2) {};
%     \node [below right=0em and -0.8em of cha2] (system2) {$\mathtt{serv}_2$//($H_1$)};
%     \node [right=1em of cha2] (b12) {};
    \node [draw, circle, right=16em of channel3] (b23) {$\cH_3$};
%     \node [below=0em of b2] (system2) {$\mathtt{serv}_3$//($H_2$)};
    
    % Draw edges
    \path [line, <->] (channel3) --node {$(-1)^{ab+bc+da} (\mathsf{I}\otimes \mathsf{W}(c,d)) |y\rangle$} (b23);
\end{tikzpicture}
\caption{Change of states in the quantum teleportation protocol with an operation $\mathsf{W}(c,d)$ on $\cH_3$ (Protocol \ref{prot:qtr}).
The symbol $\mathbf{M}_{\mathbb{Z}_2^2}\mapsto (a,b)$ implies that the PVM $\mathbf{M}_{\mathbb{Z}_2^2}$ is applied on $\cH_1\otimes \cH_2$ and the measurement outcome is $(a,b)\in\mathbb{Z}_2^2$.
}
\label{fig:qtrevised}
\end{center}
\end{figure}

\subsection{Quantum Teleportation with an Operation}
First, we give a modified version of the quantum teleportation protocol \cite{BBCJPW93},
where an operation $\mathsf{W}(c,d)$ is performed on $\cH_3$ before the quantum teleportation protocol starts.

% Consider the following protocol where $\mathsf{W}(c,d)$ is performed on $\cH_3$ before the quantum teleportation protocol starts.
\begin{prot}    \label{prot:qtr}
Suppose that Alice possesses two qubits $\cH_1$ and $\cH_2$,
Bob possesses a qubit $\cH_3$.
The state on $\cH_1$ is $\rho$ and Alice and Bob share $|\Phi\rangle \in \cH_2\otimes \cH_3$.
% Consider sending any state $\rho$ on $\cH_1$ from Alice to Bob with preserving the entanglement of $\rho$.
% This task can be done by local quantum operations and classical transmission of two bits.
Quantum teleportation protocol with an operation is given as follows.
\begin{enumerate}[Step 1]
\item Bob applies the unitary operation $\mathsf{W}(c,d)$ on $\cH_3$.
% \item Alice and Bob applies the quantum teleportation protocol.
\item Alice applies PVM $\mathbf{M}_{\mathbb{Z}_2^2}$ on $\cH_1\otimes \cH_2$ and sends the measurement outcome $(a,b)$ to Bob.
\item Bob applies the unitary $\mathsf{W}(a,b)$ on $\cH_3$.
\end{enumerate}
\end{prot}
The resultant state on $\cH_3$ is $\mathsf{W}(c,d) \rho\mathsf{W}(c,d)^{\dagger}$ and it preserves the entanglement.
Note that Protocol \ref{prot:qtr} requires two-bit transmission from Alice to Bob.
The protocol without Step 1 in Protocol \ref{prot:qtr} is the quantum teleportation protocol \cite{BBCJPW93}.

\subsubsection{Analysis of Protocol \ref{prot:qtr}}
We show that the resultant state on $\cH_3$ is $\mathsf{W}(c,d) \rho\mathsf{W}(c,d)^{\dagger}$ and it preserves the entanglement (see Fig.~\ref{fig:qtrevised}).

Let $\cH_0$ be a qubit
and $|y\rangle = \sum_{i,j=0}^{1} y_{ij}|i,j\rangle\in\cH_0\otimes \cH_1$ be a purification of the state $\rho$.
% In the following, we show that $|y\rangle$ is transmitted from $\cH_0\otimes\cH_1$ to $\cH_0\otimes\cH_3$, which implies that the state $\rho$ is transmitted with entanglement.
Before the protocol starts, the state on $\cH_0\otimes \cH_1\otimes \cH_2\otimes \cH_3$ is 
\begin{align}
|z\rangle \coloneqq \frac{1}{\sqrt{2}}\sum_{i,j,r=0}^{1} y_{ij} | i,j,r,r \rangle.     \label{eq:z}
\end{align}
% After Step 0, the state on $\cH_0\otimes \cH_1\otimes \cH_2\otimes \cH_3$ is 
% % \begin{align}
% % |z\rangle \coloneqq \frac{1}{\sqrt{2}}\sum_{i,j,r=0}^{1}  y_{ij} | i,j,r,r \rangle.
% % \end{align}
% \begin{align}
% |z'\rangle &\coloneqq  (\mathsf{I}_{\cH_0\otimes\cH_1\otimes\cH_2} \otimes \mathsf{W}(c,d)) |z\rangle   \\
%             &= \frac{1}{\sqrt{2}}\sum_{i,j,r=0}^{1}  (-1)^{rd} y_{ij} | i,j,r,r+c \rangle,
% \end{align}
% % where $|z\rangle$ is defined in \eqref{eq:z}
If the measurement outcome is $(a,b)$ in Step 2, 
	the state on $\cH_0\otimes \cH_3$ at the end of Step 2 is 
\begin{align}
 & 2 (\mathsf{I}_{\cH_0}\otimes [ \langle \Phi | (  \mathsf{I}_{\cH_1} \otimes \mathsf{W}(a,b) )^{\dagger}] \otimes \mathsf{W}(c,d)) |z\rangle  \label{eq:afterStep1}\\
 &= 2  (\mathsf{I}_{\cH_0} \otimes \langle \Phi| \otimes \mathsf{I}_{\cH_3}) \nonumber \\
	  &\quad\enskip \cdot (\mathsf{I}_{\cH_0}\otimes  \mathsf{I}_{\cH_1} \otimes \mathsf{W}(a,b)^{\dagger} \otimes \mathsf{W}(c,d)) |z\rangle  \nonumber\\
 &= 2 (-1)^{ab} (\mathsf{I}_{\cH_0} \otimes \langle \Phi| \otimes \mathsf{I}_{\cH_3}) \nonumber \\
	  &\quad\enskip \cdot (\mathsf{I}_{\cH_0}\otimes  \mathsf{I}_{\cH_1} \otimes \mathsf{W}(a,b) \otimes \mathsf{W}(c,d)) |z\rangle \label{eq:Comment..} \\
 %& 2 \cdot (\mathsf{I}_{\cH_0}\otimes \langle \Phi | (  \mathsf{I}_{\cH_1} \otimes \mathsf{W}(a,b) )^{\dagger}\otimes \mathsf{W}(c,d))  |z\rangle  \label{eq:afterStep1}\\
 &= \sum_{i,j=0}^{1} y_{ij}  (-1)^{jb+jd+ad} | i, j+a+c \rangle \nonumber \\
 &= (-1)^{ad}\sum_{i,j=0}^{1} y_{ij}  (-1)^{j(b+d)} | i, j+a+c \rangle \nonumber\\
 &= (-1)^{ad} (\mathsf{I}_{\cH_0}\otimes \mathsf{W}(a+c,b+d)) |y\rangle \nonumber\\
 &= (\mathsf{I}_{\cH_0}\otimes \mathsf{W}(c,d)\mathsf{W}(a,b)) |y\rangle, \label{eq:before_recovery}
\end{align}
where the multiplicand $2$ in \eqref{eq:afterStep1} is the normalizing multiplicand,
	\eqref{eq:Comment..} is from \eqref{property0},
	and \eqref{eq:before_recovery} is from \eqref{eq:sum}.
	%the square root of the probability for measuring $(a,b)\in\mathbb{Z}_2^2$.
At the end of Step 3, the state on $\cH_0\otimes \cH_3$ is 
\begin{align}
%     \lefteqn{(-1)^{ad+bc+}  \sum_{i,j=0}^{1} y_{ij}  (-1)^{jd} | i, j+c \rangle }   \\
%  &= 
&(\mathsf{I}_{\cH_0}\otimes \mathsf{W}(a,b)\mathsf{W}(c,d)\mathsf{W}(a,b)) |y\rangle, \nonumber \\
&= (-1)^{bc+ad} (\mathsf{I}_{\cH_0}\otimes \mathsf{W}(c,d) \mathsf{W}(a,b)^2) |y\rangle, \label{last2}\\
&= (-1)^{ab+bc+ad} (\mathsf{I}_{\cH_0}\otimes \mathsf{W}(c,d)) |y\rangle,  \label{last3}\\
&= (-1)^{ab+bc+ad}  (\mathsf{I}_{\cH_0}\otimes \mathsf{W}(c,d)) |y\rangle,    \label{eq:after_recovery}
\end{align}
where \eqref{last2} is from \eqref{eq:comm2}
	and \eqref{last3} is from \eqref{eq:sum}.
Eq.~\eqref{eq:after_recovery} is an identical state to $(\mathsf{I}_{\cH_0}\otimes \mathsf{W}(c,d)) |y\rangle$.
Therefore, the resultant state on $\cH_3$ is $\mathsf{W}(c,d) \rho\mathsf{W}(c,d)^{\dagger}$ and it preserves the entanglement.

% \noindent\textit{Case 2 (Steps 1,0,2): \qquad}
\begin{remark}
Even in case that the order of Step 1 and Step 2 is reversed, 
% i.e., when $\mathsf{W}(c,d)$ is applied after the PVM and before the recovery operation $\mathsf{W}(-a,b)$, 
the state before and after the operation $\mathsf{W}(a,b)$ is identical to \eqref{eq:before_recovery} and \eqref{eq:after_recovery}.
\end{remark}

\subsection{Two-Sum Transmission Protocol}
Consider there are three parties Alice, Bob, and Carol.
By the following protocol, Carol receives the sum of Alice's information $(a,b)\in\mathbb{Z}_2^2$ and Bob's information $(c,d)\in\mathbb{Z}_2^2$.

\begin{prot} \label{prot:add}
Suppose that the joint state of two qubits $\cH_1$ and $\cH_2$ is the maximally entangled state $|\Phi\rangle$
	and Alice and Bob possess $\cH_1$ and $\cH_2$, respectively.
The two-sum transmission protocol is given as follows.
\begin{enumerate}[Step 1]
% \item Let Alice and Bob share the maximally entangled state $|\Phi\rangle\in\cH_1\otimes\cH_2$.
\item Alice and Bob apply ${\mathsf{W}(a,b)}$ on $\cH_1$ and $\mathsf{W}(c,d)$ on $\cH_2$, respectively.
\item Alice and Bob send the quantum systems $\cH_1$ and $\cH_2$ to Carol, respectively.
\item Carol performs the PVM $\mathbf{M}_{\mathbb{Z}_\ell^2}$ and obtains the measurement outcome $(e,f)$ as the protocol output.
\end{enumerate}
\end{prot}
In Protocol \ref{prot:add}, the output $(e,f)$ is $(a+c,b+d)$,
which can be proved trivially from \eqref{eq:sum} and \eqref{eq:exchange}.
The protocol requires two-qubit transmission each from Alice and Bob.
% in the same way as the QPIR protocol in \cite{SH19}.

\section{Symmetric QPIR Protocol with $\vN-1$ Colluding Servers} \label{sec:protocol}
In this section, we propose a $(\vN-1)$-private QSPIR protocol that achieves the performance given in Theorem \ref{theo:main2} for any $\vN\geq2$ servers.
% The proposed protocol extends a $(\vN-1)$-colluded CPIR protocol by the idea of the entanglement swapping \cite{BBCJPW93}.
% Our QPIR protocol 
% requires the file size $\vL = 2^{2\ell}$ for any $\ell\ge 1$, 
% keeps the perfect server and user secrecy, and 
% achieves the QPIR rate $\lceil \vN/2 \rceil^{-1}$.
% Moreover, the sequence $\{\Psi^{(\vL_{\ell})}\}_{\ell=1}^{\infty}$ of our protocols for $\vL_{\ell} \coloneqq 2^{2\ell}$
% achieves the negligible upload cost with respect to the download cost.
In our protocol, each server contains the following file set.
Given two arbitrary integers $\ell\geq 1$ and $\vF\geq 2$,
the file set is given by the collection of $2\ell$-bit files $W_1,\ldots,W_{\vF}
\in \mathbb{Z}_2^{2\ell}$.
Each file $W_i$ is denoted by 
$$W_i = (W_i^{(1)},\ldots, W_i^{(\ell)})\in (\mathbb{Z}_2^2)^{\times \ell}.$$

Section \ref{prot3} presents our $(\vN-1)$-private QSPIR protocol with three servers ($\vN = 3$) and $\ell=1$ as the simplest case.
Then, by using Protocol \ref{prot:add} and the idea of the protocol in Section \ref{prot3},
Section \ref{protn} presents our protocol for any $\vN$ servers and any $\ell$.

\subsection{Construction of Protocol for $\vN=3$ and $\ell=1$} \label{prot3}

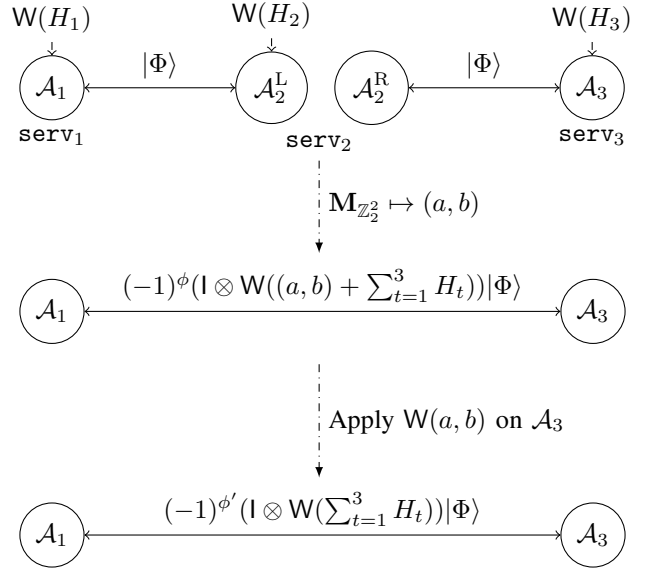
\begin{figure}[t]
\begin{center}
\begin{tikzpicture}[node distance = 3.3em, every text node part/.style={align=center}, auto, scale=0.7]
    % Place nodes
    \node [draw, circle] (channel) {$\cH_1$};
    \node [below=0em of channel] (system) {$\mathtt{serv}_1$};
    \node [draw, circle, right=5.7em of channel] (cha) {$\cH_2^{\mathrm{L}}$};
    \node [below right=0.5em and -0.8em of cha] (alice) {$\mathtt{serv}_2$};
    \node [draw, circle, right=1em of cha] (b1) {$\cH_2^{\mathrm{R}}$};
    \node [draw, circle, right=5.7em of b1] (b2) {$\cH_3$};
    \node [below=0em of b2] (system) {$\mathtt{serv}_3$};
    \node [above=0.5em of channel] (op1) {${\mathsf{W}(H_1)}$};
    \node [above=0.5em of cha] (op2) {$\mathsf{W}(H_2)$};
%     \node [above=0em of cha] (op2) {$\mathsf{W}(H_2)$};
%     \node [above=0em of b1]
    \path [line, ->,dashed] (op1) -- (channel);
    \path [line, ->,dashed] (op2) -- (cha);
    \node [above=0.5em of b2] (op3) {$\mathsf{W}(H_3)$};
    \path [line, ->,dashed] (op3) -- (b2);

    % Draw edges
    \path [line, <->] (channel) --node {$|\Phi\rangle$} (cha);
    \path [line, <->] (b1) --node {$|\Phi\rangle$} (b2);
%     \path [line] (decoder) --node {$\mathcal{D}(\sigma)$} (5,0);

    \node [below=3.5em of alice] (measure) {};
    \path [line, =>,dashdotted] (alice) --node {$\mathbf{M}_{\mathbb{Z}_2^2}\mapsto (a,b)$} (measure);

    \node [draw, circle, below=6em of channel] (channel2) {$\cH_1$};
%     \node [below=0em of channel2] (system2) {$\mathtt{serv}_1$//($H_3$)};
%     \node [right=5em of channel2] (cha2) {};
%     \node [below right=0em and -0.8em of cha2] (system2) {$\mathtt{serv}_2$//($H_1$)};
%     \node [right=1em of cha2] (b12) {};
    \node [draw, circle, right=18em of channel2] (b22) {$\cH_3$};
%     \node [below=0em of b2] (system2) {$\mathtt{serv}_3$//($H_2$)};
   
    % Draw edges
    \path [line, <->] (channel2) --node {$(-1)^{\phi}(\mathsf{I}\otimes\mathsf{W}((a,b)+\sum_{t=1}^3 H_t))|\Phi\rangle$} (b22);
%     \path [line, <->] (b12) --node {$|\Phi\rangle$} (b22);

    \node [below=7em of alice] (x) {};
    \node [below=12em of alice] (step2) {};
    \path [line, =>,dashdotted] (x) --node {Apply $\mathsf{W}(a,b)$ on $\cH_3$} (step2);

    \node [draw, circle, below=6em of channel2] (channel3) {$\cH_1$};
%     \node [below=0em of channel2] (system2) {$\mathtt{serv}_1$//($H_3$)};
%     \node [right=5em of channel2] (cha2) {};
%     \node [below right=0em and -0.8em of cha2] (system2) {$\mathtt{serv}_2$//($H_1$)};
%     \node [right=1em of cha2] (b12) {};
    \node [draw, circle, right=18em of channel3] (b23) {$\cH_3$};
%     \node [below=0em of b2] (system2) {$\mathtt{serv}_3$//($H_2$)};
    
    % Draw edges
    \path [line, <->] (channel3) --node {$(-1)^{\phi'}(\mathsf{I}\otimes\mathsf{W}(\sum_{t=1}^3 H_t))|\Phi\rangle$} (b23);
\end{tikzpicture}
\caption{Two-private QSPIR protocol for three servers and $\ell=1$.
$\mathbf{M}_{\mathbb{Z}_2^2}\mapsto (a,b)$ implies that the PVM $\mathbf{M}_{\mathbb{Z}_2^2}$ is applied on $\cH_2^{\mathrm{L}}\otimes \cH_2^{\mathrm{R}}$ and the measurement outcome is $(a,b)\in\mathbb{Z}_2^2$.
The values $\phi, \phi'\in\mathbb{Z}_2$ are determined by $(a,b)$, $H_1$, $H_2$, and $H_3$.
}
\label{fig:three}
\end{center}
\end{figure}

\subsubsection{Protocol}
Our protocol for three servers each containing the file set $W_1,\ldots,W_\vF \in \mathbb{Z}_2^2$ is described as follows (see Fig.~\ref{fig:three}).
\begin{enumerate}[Step 1]
\item 
% Let $\cH_1$, $\cH_2^{\mathrm{L}}$, $\cH_2^{\mathrm{R}}$, $\cH_{3}$ be qubits.
% The server $\mathtt{serv}_1$ possesses a qubit $\cH_1$,
% $\mathtt{serv}_2$ possesses two qubits $\cH_2^{\mathrm{L}}$ and $\cH_2^{\mathrm{R}}$,
% and $\mathtt{serv}_3$ posesses a qubit $\cH_3$.
The servers $\mathtt{serv}_1$, $\mathtt{serv}_2$, $\mathtt{serv}_3$ possess one qubit $\cH_1$, {two qubits} $\cH_2^{\mathrm{L}}$, $\cH_2^{\mathrm{R}}$, %$\cH_2^{\mathrm{M}}$, and $\cH_2^{\mathrm{M}2}$,
and one qubit $\cH_3$, respectively.
The initial states on both of $\cH_1\otimes \cH_2^{\mathrm{L}}$ and $\cH_2^{\mathrm{R}}\otimes \cH_3$ are the maximally entangled state $|\Phi\rangle$.

\item 
	Let $K$ be the index of the file to be retrieved.
	Choose two subsets $Q_1$ and $Q_2$ of $\{1,\ldots,\vF\}$ independently and uniformly at random.
Define $Q_3$ by
$$
Q_3 \coloneqq Q_1 \oplus Q_2 \oplus \{K\}.
$$
%\myred{where $K$ is the index of the file to be retrieved - compare the description at Section~\ref{sec:problem}.}
For each $t\in\{1,2,3\}$, the user sends the query $Q_t$ to $\mathtt{serv}_t$.

\item For each $t\in\{1,2,3\}$, the server $\mathtt{serv}_t$ calculates
\begin{align}
   H_t \coloneqq \sum_{i\in Q_t} W_i. 
\end{align}
The server $\mathtt{serv}_1$ ($\mathtt{serv}_3$) applies
    ${\mathsf{W}(H_1)}$ to $\cH_1$ 
    ($\mathsf{W}(H_3)$ to $\cH_3$)
    and transmits $\cH_1$ ($\cH_3$) to the user.
The server $\mathtt{serv}_2$ 
applies $\mathsf{W}(H_2)$ on $\cH_2^{\mathrm{L}}$,
performs the PVM 
$\mathbf{M}_{\mathbb{Z}_2^2}$ on $\cH_2^{\mathrm{L}} \otimes \cH_2^{\mathrm{R}}$,
% $$\{ (\mathsf{I} \otimes \mathsf{W}(a,b))|\Phi \rangle  : a,b\in\mathbb{Z}_2 \}$$
and transmits the measurement outcome $(a,b)\in \mathbb{Z}_2^2$ to the user. 
%\myred{Since we only consider quantum communication from the servers to the user,
%	by sending $|(a,b)\rangle\in\cH_2^{\mathrm{M}}\otimes\cH_2^{\mathrm{M}2}$.}

\item 
% The user performs the following process.
The user applies $\mathsf{W}(a,b)$ on $\cH_3$ and
performs the PVM 
$\mathbf{M}_{\mathbb{Z}_2^2}$
% $\{ (\mathsf{I} \otimes \mathsf{W}(c,d))|\Phi \rangle  : c,d\in\mathbb{Z}_2 \}$ 
on $\cH_1 \otimes \cH_3$,
and the output of the protocol is the measurement outcome $\hat{W}_K \in \mathbb{Z}_2^2$.

\end{enumerate}

\subsubsection{Analysis}

First, we show the correctness of the protocol.
The state of $\cH_1\otimes \cH_2^{\mathrm{L}}\otimes \cH_2^{\mathrm{R}} \otimes \cH_3$ before the PVM at Step 3 is 
\begin{align}
	&({\mathsf{W}(H_1)}\otimes \mathsf{W}(H_2) )|\Phi\rangle \otimes 
	(\mathsf{I}\otimes \mathsf{W}(H_3) ) |\Phi\rangle \\
	&=
	(-1)^{\phi_0} ({\mathsf{W}(H_1+H_2)}\otimes \mathsf{I} )|\Phi\rangle \otimes 
	(\mathsf{I}\otimes \mathsf{W}(H_3) ) |\Phi\rangle
\end{align}
by \eqref{eq:sum},
	where $\phi_0\in\mathbb{Z}_2$ is determined depending on $H_1$ and $H_2$.
After the PVM at Step 3 with the measurement outcome $(a,b)$, the state on $\cH_1 \otimes \cH_3$ is 
	\begin{align}
	&\paren*{{\mathsf{W}(H_1+H_2)}\otimes 
		[ \langle \Phi | (  \mathsf{I}_{\cH_1} \otimes \mathsf{W}(a,b) )^{\dagger}] \otimes 
			\mathsf{W}(H_3) } |\Phi\rangle \otimes |\Phi \rangle  \nonumber \\
	&= (-1)^{\phi_{0}'} \paren*{{\mathsf{W}(H_1+H_2)}\otimes  \langle \Phi |  \otimes 
			\mathsf{W}((a,b)+H_3)} |\Phi\rangle \otimes |\Phi \rangle  \nonumber \\
	&=(-1)^{\phi_{0}'} \paren*{\mathsf{W}(H_1+H_2)\otimes \mathsf{W}((a,b)+H_3) } |\Phi\rangle \nonumber  \\
	&=
	(-1)^{\phi_0''} (\mathsf{I}\otimes \mathsf{W}((a,b)+H_1+H_2+H_3) ) |\Phi\rangle ,
	\label{temnMP}
\end{align}
where $\phi_0',\phi_0''\in\mathbb{Z}_2$ are determined by $(a,b), H_1,H_2,H_3$.
Thus, after the user's operation at Step 4, 
	the state on $\cH_1 \otimes \cH_3$ is 
	\begin{align}
	(\mathsf{I}\otimes \mathsf{W}(H_1+H_2+H_3) ) |\Phi\rangle .
	\label{temnP}
	\end{align}
(Alternatively, %comparing Fig.~\ref{fig:qtrevised} and Fig.~\ref{fig:three},
	we can also obtain the same result \eqref{temnP} by considering the servers and the user apply Protocol \ref{prot:qtr} for $(c,d)\coloneqq H_3$ and $|y\rangle \coloneqq ( {\mathsf{W}(H_1)}\otimes \mathsf{W}(H_2) )|\Phi\rangle = (-1)^{\phi_0}(\mathsf{I}\otimes \mathsf{W}(H_1+H_2)) |\Phi\rangle$.)
Therefore, the user obtains the measurement outcome $\hat{W}_K = \sum_{t=1}^3 H_t=W_K$, which implies the correctness of our protocol.

%% Note that the order of the servers in applying operations can be arbitrary since the operations of each server are applied on different quantum systems.
%%Note that the application of ${\mathsf{W}(H_1)}$ and $\mathsf{W}(H_2)$ corresponds to Protocol \ref{prot:add},
%%and 
%the remaining process corresponds to Protocol \ref{prot:qtr} for $(c,d)\coloneqq H_3$ and $|y\rangle \coloneqq ( {\mathsf{W}(H_1)}\otimes \mathsf{W}(H_2) )|\Phi\rangle = (-1)^{\phi_0}(\mathsf{I}\otimes \mathsf{W}(H_1+H_2)) |\Phi\rangle$ where $\phi_0\in\mathbb{Z}_2$ is determined depending on $H_1$ and $H_2$.
%Therefore, 
%% as given in Fig. \ref{fig:three},
%% the state on $\cH_1\otimes \cH_3$ after Step 3 is 
%% $$(-1)^{\phi''}(\mathsf{I}\otimes\mathsf{W}((a,-b)+\sum_{t=1}^3 H_t))|\Phi\rangle,$$
%% where $\phi''\in\mathbb{Z}_2$ is determined by $(a,b), H_1, H_2, H_3$.
%% Moreover, 
%after the operation $\mathsf{W}(a,b)$ in Step 4, 
%the resultant state on $\cH_1\otimes \cH_3$ is
%$$(-1)^{\phi'}\paren*{\mathsf{I}\otimes\mathsf{W}\paren*{\sum_{t=1}^3 H_t}}|\Phi\rangle,$$
%where $\phi'\in\mathbb{Z}_2$ is determined by $(a,b)$, $H_1$, $H_2$, and $H_3$.
%Thus, the measurement outcome in Step 4 is $\sum_{t=1}^3 H_t=W_K$, which implies the correctness of our protocol.

The user secrecy follows from the fact that any two of $H_1,H_2,H_3$ are independent of the query index $K$.
The server secrecy is proved as follows. 
	The information the user obtains is the measurement outcome $(a,b)$ and the state \eqref{temnMP}.
	The measurement outcome $(a,b)$ is uniformly at random because the state on $\cH_2^{\mathrm{L}}\otimes \cH_2^{\mathrm{R}}$ is the completely mixed state before the measurement $\mathbf{M}_{\mathbb{Z}_2^2}$ at Step 3.
	Since the state \eqref{temnMP} only depends on $W_K$ and $(a,b)$, which are jointly independent of the files other than $W_K$, 
		the user obtains no information of the files other than $W_K$.
%		
%		it is enough to show that $(a,b)$ is independent of the files other than $W_K$.
%	Notice that the state on $\cH_2^{\mathrm{L}}\otimes \cH_2^{\mathrm{R}}$ is the completely mixed state.
%		the measurement outcome $(a,b)\in\mathbb{Z}_2^2$ is uniformly at random.}
%	from the fact that the user's information is $W_K, a,b$ which is independent of any file except for $W_K$.

The upload cost is $\vN\vF=3\vF$ bits because each of $Q_1, Q_2, Q_3$ is written by $\vF$ bits.
In the protocol, the user downloads $2$ qubits and $2$ bits %because two qubits $\cH_1$, $\cH_3$ are downloaded and the classical information $(a,b)$ can be transmitted by $2$ qubits.
	but we count the download cost as $4$ qubits 
	since we only count quantum communication in our QPIR model (Section~\ref{sec:problem})
	and 
	one qubit can convey one bit at most.
%	
%	the download cost is $4$ qubits because two qubits $\cH_1$, $\cH_3$ are downloaded and the classical information $(a,b)$ can be transmitted by $2$ qubits.
% four qubits are downloaded which consist of two qubits for $\cH_1, \cH_3$ and two qubits for downloading the measurement outcome $(a,b)$.
The file size is $2\ell=2$ bits.
Therefore, 
% the upload cost is negligible with respect to the download cost, and
the QPIR rate is 
$2/(\vN+1) = 2/4$.
% which is strictly greater than 
% the $(\vN-1)$-private CPIR capacity $1/\vN=1/3$ when any $\vN-1$ servers may collude and the number of files are infinite.

\subsection{Construction of Protocol for $\vN$ Servers} \label{protn}
% Suppose each of $\vN$ servers contains the collection of $2\ell$-bit files $W_1,\ldots,W_{\vF} \in \{0,\ldots, \vL-1\}$.
% % when any $\vN-1$ servers possibly collude.
% The set $\{0,\ldots, \vL_\ell-1\}$ is identified with $(\mathbb{Z}_2^2)^{\times \ell}$ and
% $W_i$ for any $i\in\{1,\ldots,\vF\}$ is denoted by 
% $W_i = (W_i^{(1)},\ldots, W_i^{(\ell)})\in (\mathbb{Z}_2^2)^{\times \ell}$.

In this subsection, we present our protocol for any $\vN\geq 2$ servers and any $\ell\geq 1$.
The idea of our protocol construction is described as follows.
The number of servers $\vN$ are generalized to be arbitrary by using the idea of the three-server protocol in Section \ref{prot3}.
In this generalization, it is necessary for servers to transmit the sum of measurement outcomes to the user,
and it is performed efficiently by using the two-sum transmission protocol (Protocol \ref{prot:add}).
The index $\ell$ is increased by using the same query repetitively until the protocol retrieves the entire file information.

Our protocol for $\vN$ servers is described as follows (see Fig. \ref{fig:protocol}).

%Given two arbitrary integers $\ell\geq 1$ and $\vF\geq 2$,
%the file set is given by the collection of $2\ell$-bit files $W_1,\ldots,W_{\vF}
%\in \mathbb{Z}_2^{2\ell}$.
%Each file $W_i$ is denoted by 
%$W_i = (W_i^{(1)},\ldots, W_i^{(\ell)})\in (\mathbb{Z}_2^2)^{\times \ell}$.

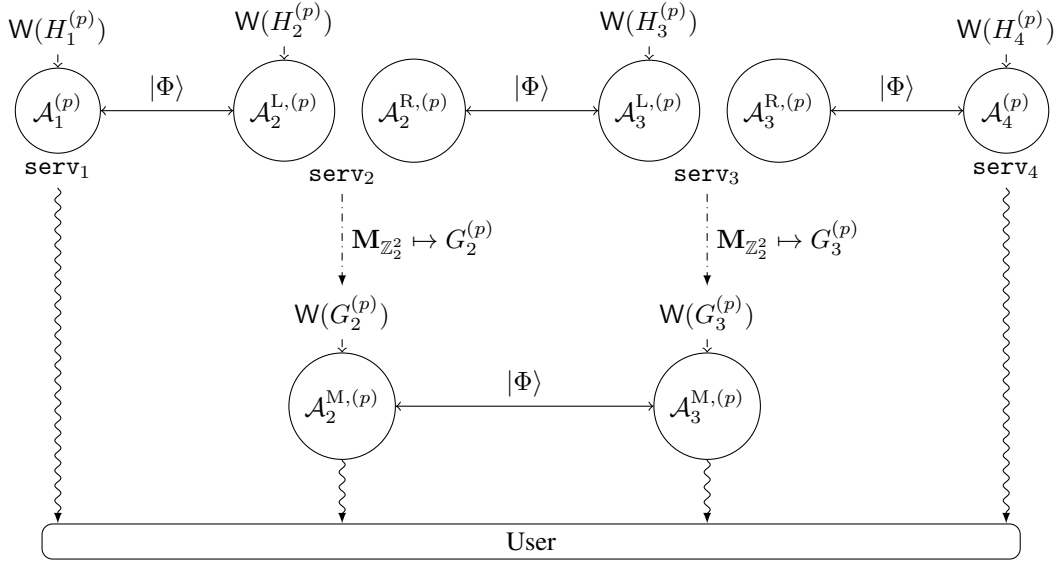
\begin{figure*}[th] 
\begin{center}
\begin{tikzpicture}[node distance = 3.3em, every text node part/.style={align=center}, auto, scale=0.7]
    \node [draw, circle] (channel) {$\cH_1^{(p)}$};
    \node [below=0em of channel] (system) {$\mathtt{serv}_1$};
    \node [above= 0.5em of channel] (H1) {${\mathsf{W}(H_1^{(p)})}$};

    \node [draw, circle, right=5em of channel] (cha) {$\cH_{2}^{\mathrm{L},(p)}$};
    \node [above= 0.5em of cha] (H2) {$\mathsf{W}(H_2^{(p)})$};
    \node [below right=0.5em and -0.8em of cha] (alab) {$\mathtt{serv}_2$};
    \node [draw, circle, right=1em of cha] (b1) {$\cH_{2}^{\mathrm{R},(p)}$};
    \node [draw, circle, below= 6em of alab] (b0) {$\cH_{2}^{\mathrm{M},(p)}$};
    \node [above= 0.5em of b0] (G2) {${\mathsf{W}(G_2^{(p)})}$};
    \path [line, =>,dashdotted] (alab) --node {$\mathbf{M}_{\mathbb{Z}_2^2} \mapsto G_2^{(p)}$} (G2);

    \node [draw, circle, right=5em of b1] (b2) {$\cH_{3}^{\mathrm{L},(p)}$};
    \node [above= 0.5em of b2] (H3) {$\mathsf{W}(H_3^{(p)})$};
    \node [below right=0.5em and -0.8em of b2] (blab) {$\mathtt{serv}_3$};
    \node [draw, circle, right=1em of b2] (c1) {$\cH_{3}^{\mathrm{R},(p)}$};
    \node [above= 1em of c1] (H3p) {};
    \node [draw, circle, below=6em of blab] (c0) {$\cH_{3}^{\mathrm{M},(p)}$};
    \node [above= 0.5em of c0] (G3) {$\mathsf{W}(G_3^{(p)})$};
    \path [line, =>,dashdotted] (blab) --node {$\mathbf{M}_{\mathbb{Z}_2^2} \mapsto G_3^{(p)}$} (G3);

    \node [draw, circle, right=5em of c1] (c2) {$\cH_{4}^{(p)}$};
    \node [below=0em of c2] (clab) {$\mathtt{serv}_4$};
    \node [above= 0.5em of c2] (H4) {${\mathsf{W}(H_4^{(p)})}$};

    \node [draw, block, minimum width = 37em, below left=3em and -13.3em of c0]  (user) {User};
    
    \path [line,snake arrow] (system.south) -- (user.north -| channel);
    \path [line,snake arrow] (clab.south) -- (user.north -| c2);
    \path [line,snake arrow] (b0.south) -- (user.north -| b0);
    \path [line,snake arrow] (c0.south) -- (user.north -| c0);
    
    % Draw edges
    \path [line, <->] (channel) --node {$|\Phi\rangle$} (cha);
    \path [line, <->] (b1) --node {$|\Phi\rangle$} (b2);
    \path [line, <->] (c1) --node {$|\Phi\rangle$} (c2);
    \path [line, <->] (b0) --node {$|\Phi\rangle$} (c0);
    
    \path [line,->,dashed] (H1) -- (channel);
    \path [line,->,dashed] (H4) -- (c2);
    \path [line,->,dashed] (H2) -- (cha);
    \path [line,->,dashed] (H3) -- (b2);
    \path [line,->,dashed] (G2) -- (b0);
    \path [line,->,dashed] (G3) -- (c0);
    
%     \path [line] (decoder) --node {$\mathcal{D}(\sigma)$} (5,0);
\end{tikzpicture}
\caption{Download step of our $(\vN-1)$-private QSPIR protocol for four servers and any integer $1\leq p \leq \ell$.
For any $t\in\{2,3\}$,
$\mathbf{M}_{\mathbb{Z}_2^2}\mapsto G_t^{(p)}$ implies that the PVM $\mathbf{M}_{\mathbb{Z}_2^2}$ is applied on $\cH_{t}^{\mathrm{L},(p)}\otimes \cH_t^{\mathrm{R},(p)}$ and the measurement outcome is $G_t^{(p)}\in\mathbb{Z}_2^2$.
The snake shape arrow indicates the transmission of a qubit. 
}
\label{fig:protocol}
\end{center}
\end{figure*}

\subsubsection{Preparation}
For each $p\in\{1,\ldots,\ell\}$, prepare the following quantum systems and states.
The servers $\mathtt{serv}_1$ and $\mathtt{serv}_{\vN}$ have qubits $\cH_1^{(p)}$ and $\cH_{\vN}^{(p)}$, respectively.
For each $t\in\{2,\ldots,\vN-1\}$, the server $\mathtt{serv}_t$ has three qubits $\cH_{t}^{\mathrm{L},(p)}$, $\cH_{t}^{\mathrm{R},(p)}$, $\cH_{t}^{\mathrm{M},(p)}$.
If $\vN$ is odd, we consider the server $\mathtt{serv}_{\vN-1}$ has only two qubits $\cH_{\vN-1}^{\mathrm{L},(p)}$, $\cH_{\vN-1}^{\mathrm{R},(p)}$.
The maximally entangled state $|\Phi\rangle$ is shared between 
each of pairs $(\cH_1, \cH_{2}^{\mathrm{L},(p)})$, $(\cH_2^{\mathrm{R},(p)}, \cH_{3}^{\mathrm{L},(p)})$, $(\cH_3^{\mathrm{R},(p)}, \cH_{4}^{\mathrm{L},(p)})$, \ldots\ , $(\cH_{\vN-2}^{\mathrm{R},(p)}, \cH_{\vN-1}^{\mathrm{L},(p)})$, $(\cH_{\vN-1}^{\mathrm{R},(p)}, \cH_{\vN})$
and $(\cH_{2j}^{\mathrm{M},(p)}, \cH_{2j+1}^{\mathrm{M},(p)})$ for any $j\in\{1,\ldots,\lfloor \vN/2\rfloor-1\}$.

\subsubsection{Upload Step}

Let $K$ be the index of the file to be retrieved.
Choose subsets $Q_1, \ldots, Q_{\vN-1}$ of $\{1,\ldots,\vF\}$ independently and uniformly at random.
Define $Q_{\vN}$ by
$$
Q_{\vN} \coloneqq \bigoplus_{t=1}^{\vN-1} Q_t \oplus \{K\}.
$$
The user sends the query $Q_t$ to $\mathtt{serv}_t$ for each $t\in\{1,\ldots,\vN\}$.

\subsubsection{Download Step}
For each $t\in\{1,\ldots,\vN\}$,
depending on the query $Q_t$, the server $\mathtt{serv}_t$ calculates
\begin{align}
   \lefteqn{H_t = (H_t^{(1)},\ldots,H_t^{(\ell)})}\nonumber \\ 
 &\coloneqq \sum_{i\in Q_t} W_i = ( \sum_{i\in Q_t} W_i^{(1)} , \ldots, \sum_{i\in Q_t} W_i^{(\ell)} ).
% H_i = \sum_{i\in Q_i} W_i.
\end{align}
Then, for each $p\in\{1,\ldots,\ell\}$, the servers perform the following process.
\begin{enumerate}[a)]
\item The server $\mathtt{serv}_1$ ($\mathtt{serv}_{\vN}$) applies
    ${\mathsf{W}(H_1^{(p)})}$ to $\cH_1$ ($\mathsf{W}(H_{\vN}^{(p)})$ to $\cH_{\vN}^{(p)}$)
    and transmits $\cH_1^{(p)}$ ($\cH_{\vN}^{(p)}$) to the user.
 
\item For each $t\in\{2,\ldots,\vN-1\}$, 
the server $\mathtt{serv}_t$ applies $\mathsf{W}(H_t^{(p)})$ on $\cH_t^{\mathrm{L},(p)}$ and
performs the PVM 
$\mathbf{M}_{\mathbb{Z}_2^2}$ on $\cH_t^{\mathrm{L},(p)} \otimes \cH_t^{\mathrm{R},(p)}$
% = \{ M_{(a,b)} = (\mathsf{I} \otimes \mathsf{W}(a,b))|\Phi \rangle  : a,b\in\mathbb{Z}_2 \}$$ on $\cH_t^{\mathrm{L},(p)} \otimes \cH_t^{\mathrm{R},(p)}$ 
whose measurement outcome is denoted by $G_t^{(p)}\in\mathbb{Z}_2^2$.

\item 
For each $j\in\{1,\ldots,\lfloor \vN/2 \rfloor-1\}$,
the servers $\mathtt{serv}_{2j}$ and $\mathtt{serv}_{2j+1}$ transmit the sum $G_{2j}^{(p)}+G_{2j+1}^{(p)}$ to the user 
by the two-sum transmission protocol (Protocol \ref{prot:add}) with qubits $\cH_{2j}^{\mathrm{M},(p)}$ and $\cH_{2j+1}^{\mathrm{M},(p)}$. 

\item If $\vN$ is odd, $\mathtt{serv}_{\vN-1}$ transmits $G_{\vN-1}^{(p)}\in\mathbb{Z}_2^2$ to the user.
\end{enumerate}
\subsubsection{Retrieval Step}
For each $p\in\{1,\ldots,\ell\}$, the user performs the following process.
\begin{enumerate}[a)]
\item 
For any $j\in\{1,\ldots,\lfloor \vN/2 \rfloor-1\}$, the user receives the sum $G_{2j}^{(p)}+G_{2j+1}^{(p)}$
by Download Step c). 
% For any $j\in\{1,\ldots,\lfloor \vN/2 \rfloor-1\}$,
% the user performs the measurement $\mathbf{M}_{\mathbb{Z}_2^2}$ on $\cH_{2j}^{\mathrm{M},(p)}\otimes \cH_{2j+1}^{\mathrm{M},(p)}$
% whose measurement outcome is denoted by $F_j^{(p)}\in\mathbb{Z}_2^2$.
% % $G_{2j}+G_{2j+1}$. % = (g_{2j}^{1}+g_{2j+1}^{1},g_{2j}^{2}+g_{2j+1}^{2}).$$
If $\vN$ is odd, the user receives $G_{\vN-1}^{(p)}$ additionally.
	%is obtained from the state $|G_{\vN-1}^{(p)}\rangle$.
% % If $\vN$ is odd, let $F_{\lfloor\vN/2\rfloor} \coloneqq G_{\vN-1}$ which can be obtained from the state $|G_{\vN-1}\rangle$.

\item 
% Let $\sum_{t=2}^{\vN-1} G_t^{(p)} = (f_1^{(p)},f_2^{(p)})$. 
The user applies $\mathsf{W}(\sum_{t=2}^{\vN-1} G_{t}^{(p)})$ on $\cH_{\vN}^{(p)}$.

\item
The user performs the PVM $\mathbf{M}_{\mathbb{Z}_2^2}$ on $\cH_1^{(p)} \otimes \cH_{\vN}^{(p)}$
whose measurement outcome is denoted by $\hat{W}_K^{(p)}\in\mathbb{Z}_2^2$.
% The output of the protocol is the measurement outcome $(a,b)\in\mathbb{Z}_2^2$.
\end{enumerate}
The protocol output is $\hat{W}_K = (\hat{W}_K^{(1)},\ldots, \hat{W}_K^{(\ell)})\in (\mathbb{Z}_2^2)^{\times \ell}$.

\subsection{Analysis of Protocol for $\vN$ servers}
In this subsection, we prove the correctness and the secrecy of the protocol, and analyze the costs and rate of the protocol.

\subsubsection{Correctness}
Let $p$ be any element of $\{1,\ldots,\ell\}$.
%For any $t\in\{2,\ldots,\vN-1\}$,
%define $G_t^{(p)} \coloneqq (g_{t,1}^{(p)}, g_{t,2}^{(p)})$ from the measurement outcome $G_t^{(p)} = (g_{t,1}^{(p)}, g_{t,2}^{(p)})$ of $\mathtt{serv}_t$.
% As will be shown from the below paragraph,
As shown in the next paragraph, %Appendix \ref{append:download},
at the end of Download Step, the state on  $\cH_1^{(p)}\otimes\cH_{\vN}^{(p)}$ is 
\begin{align}
(-1)^{\phi_{\vN}^{(p)}} \paren*{\mathsf{I} \otimes \mathsf{W}\paren*{\sum_{t=1}^{\vN} H_t^{(p)} + \sum_{t=2}^{\vN-1} G_t^{(p)}}} |\Phi\rangle,
\label{eq:download}
\end{align}
where $\phi_{\vN}^{(p)}\in\mathbb{Z}_2$ is determined depending on $H_1^{(p)}$, \ldots\ , $H_{\vN}^{(p)}$, $G_2^{(p)}$, \ldots\ , $G_{\vN-1}^{(p)}$.
% At Retrieval Step b), 
% $(-f_1^{(p)},f_2^{(p)}) = - \sum_{i=2}^{\vN-1} G_i^{(p)}$ holds
% which implies that the resultant state on $\cH_1^{(p)}\otimes \cH_{\vN}^{(p)}$ is 
% \begin{align}
% (-1)^{\phi_{\vN}^{(p)'}} \paren*{\mathsf{I} \otimes \mathsf{W}\paren*{\sum_{t=1}^{\vN} H_t^{(p)}}} |\Phi\rangle,
% \end{align}
% where $\phi_{\vN}^{(p)'}\in\mathbb{Z}_2$ is determined depending on $H_1^{(p)}$, $\ldots$ , $H_{\vN}^{(p)}$, $G_2^{(p)},\ldots,G_{\vN-1}^{(p)}$.
% 
Then, at the end of Retrieval Step b),
the state on $\cH_1^{(p)}\otimes \cH_{\vN}^{(p)}$ is 
\begin{align}
(-1)^{\widetilde{\phi}_{\vN}^{(p)}} \paren*{\mathsf{I} \otimes \mathsf{W}\paren*{\sum_{t=1}^{\vN} H_t^{(p)}}} |\Phi\rangle,
\end{align}
where $\widetilde{\phi}_{\vN}^{(p)}\in\mathbb{Z}_2$ is determined depending on $H_1^{(p)}$, \ldots\ , $H_{\vN}^{(p)}$, $G_2^{(p)}$, \ldots\ , $G_{\vN-1}^{(p)}$.
Thus, 
at Retrieval Step c), 
the measurement outcome is $\hat{W}_K^{(p)} = \sum_{t=1}^{\vN} H_t^{(p)} = W_K^{(p)}\in\mathbb{Z}_2^2$.
which implies that our protocol correctly retrieves $W_K^{(p)}$.
Since $W_K^{(p)}$ is retrieved correctly for any $p$,
	the queried file $W_K = (W_K^{(1)},\ldots, W_K^{(\ell)})$ is retrieved correctly.

% In the following, we show Eq. \eqref{eq:download}.
Now, we prove \eqref{eq:download}.
Since the operations of different servers are applied on different quantum systems, 
the order of the servers' operations can be arbitrary.
Therefore, in the following, we consider that the servers $\mathtt{serv}_1,\ldots,\mathtt{serv}_{\vN}$ apply the operations sequentially.
At the end of the operation of $\mathtt{serv}_1$,
the state on $\cH_1^{(p)}\otimes \cH_{2}^{\mathrm{L},(p)}$ is
\begin{align}
|y_1\rangle\! \coloneqq \!(\mathsf{W}(H_1^{(p)}) \!\otimes\! \mathsf{I}) |\Phi\rangle
=\! (-1)^{\phi_1^{(p)}} 
    \!(\mathsf{I} \otimes \mathsf{W}(H_1^{(p)})) |\Phi\rangle,
\end{align}
where $\phi_1^{(p)}$ is determined depending on $H_1^{(p)}$.
Suppose that at the end of the operations of $\mathtt{serv}_k$ for any $k\in \{1,\ldots,\vN-2\}$,
the state on $\cH_1\otimes \cH_{k+1}^{\mathrm{L},(p)}$ is
\begin{align*}
|y_k\rangle \coloneqq (-1)^{\phi_k^{(p)}} \paren*{\mathsf{I} \otimes \mathsf{W}\paren*{\sum_{t=1}^{k} H_t^{(p)} + \sum_{t=2}^{k} G_t^{(p)}}} |\Phi\rangle,
\end{align*}
where $\phi_k^{(p)}\in\mathbb{Z}_2$ is determined depending on $H_1^{(p)}$, \ldots\ , $H_k^{(p)}$, $G_2^{(p)}$, \ldots\ , $G_{k}^{(p)}$.
Note that the operations of $\mathtt{serv}_{k+1}$ corresponds to the steps 0 and 1 of Protocol \ref{prot:qtr} for $|y\rangle\coloneqq|y_{k}\rangle$, $(a,b)\coloneqq G_{k+1}^{(p)}$, and $(c,d)\coloneqq H_{k+1}^{(p)}$.
Therefore, after the operations of $\mathtt{serv}_{k+1}$,
the state on $\cH_1^{(p)}\otimes \cH_{k+2}^{\mathrm{L},(p)}$ is 
\begin{align*}
|y_{k+1}\rangle\coloneqq(-1)^{\phi_{k+1}^{(p)}} \!\paren*{\mathsf{I} \!\otimes\! \mathsf{W}\paren*{\sum_{t=1}^{k+1} H_t^{(p)} \!+ \sum_{t=2}^{k+1} G_t^{(p)}\!}\!} |\Phi\rangle,
\end{align*}
where $\phi_{k+1}^{(p)}\in\mathbb{Z}_2$ is determined depending on $H_1^{(p)}$, \ldots\ , $H_{k+1}^{(p)}$, $G_2^{(p)}$, \ldots , $G_{k}^{(p)}$ and
the system $\cH_{k+2}^{\mathrm{L},(p)}$ denotes $\cH_{\vN}^{(p)}$  for the case $k= \vN-2$.
By the mathematical induction, 
the state on $\cH_1^{(p)}\otimes\cH_{\vN}^{(p)}$ after the operations of $\mathtt{serv}_{\vN-1}$ is 
\begin{align*}
|y_{\vN-1}\rangle = (-1)^{\phi_{\vN-1}^{(p)}} \paren*{\mathsf{I} \!\otimes\! \mathsf{W}\paren*{\sum_{t=1}^{\vN-1} H_t^{(p)} \!+ \sum_{t=2}^{\vN-1} G_t^{(p)}\!}\!} |\Phi\rangle,
\end{align*}
and after the operation of $\mathtt{serv}_{\vN}$, the state is 
\begin{align*}
(-1)^{\phi_{\vN}^{(p)}} \paren*{\mathsf{I} \otimes \mathsf{W}\paren*{\sum_{t=1}^{\vN} H_t^{(p)} + \sum_{t=2}^{\vN-1} G_t^{(p)}}} |\Phi\rangle,
\end{align*}
where $\phi_{\vN}^{(p)}\in\mathbb{Z}_2$ is determined depending on $H_1^{(p)}$, \ldots\ , $H_{\vN}^{(p)}$, $G_2^{(p)}$, \ldots\ , $G_{\vN-1}^{(p)}$.
Thus, we have Eq. \eqref{eq:download}.

\subsubsection{Secrecy}

The user secrecy is obtained because 
the collection of any $\vN-1$ variables in $Q_1,\ldots,Q_{\vN}$ is independent of the query index $K$.
Next, we consider the server secrecy.
The user obtains $W_K$ and $G_{2j}^{(p)}+G_{2j+1}^{(p)}$ for any $j\in\{1,\ldots,\lfloor\vN/2\rfloor-1\}$ and any $p\in\{1,\ldots,\ell\}$.
If $\vN$ is odd, the user obtains $G_{\vN-1}^{(p)}$ additionally.
Note that before the measurement of the qubits possessed by the server $\mathtt{serv}_t$ ($t\in\{2,\ldots,\vN-1\}$), 
	the states on $\cH_t^{\mathrm{L},(p)}$ and $\cH_t^{\mathrm{R},(p)}$
	are the completely mixed states, which implies that the measurement outcomes $G_t^{(p)}$ for all $t$ are independent of any file.
% This implies that the distribution of the user's information is independent of any file $W_1,\ldots, W_{\vF}$,
Therefore, the user obtains no file information other than $W_K$.

\subsubsection{Costs and Rate}

The upload cost is $\vN\vF$ bits because each subset $Q_1,\ldots, Q_{\vN}$ of $\{1,\ldots,\vF\}$ is written by $\vF$ bits.
For each $p\in\{1,\ldots,\ell\}$, the user 
	downloads $\vN$ qubits $\cH_1^{(p)}$, $\cH_2^{\mathrm{M},(p)}$, \ldots\ , $\cH_{\vN-1}^{\mathrm{M},(p)}$, $\cH_{\vN}^{(p)}$ if $\vN$ is even, and 
	downloads $\vN-1$ qubits $\cH_1^{(p)}$, $\cH_2^{\mathrm{M},(p)}$, \ldots\ , $\cH_{\vN-2}^{\mathrm{M},(p)}$, $\cH_{\vN}^{(p)}$ and two bits $G_{\vN-1}^{(p)}\in\mathbb{Z}_2^2$ if $\vN$ is odd.
Since we only count quantum communication in our QPIR model
	and
	one qubit can convey one bit at most, 
	the total download cost is $\vN\ell$ qubits when $\vN$ is even
	and $(\vN+1)\ell$ qubits when $\vN$ is odd.
% $4$ qubits because
% four qubits are downloaded which consist of two qubits for $\cH_1, \cH_4$ and two qubits for downloading the measurement outcome $(a,b)$.
The file size is $2\ell$ bits, i.e., $\mathsf{m} = 2^{2\ell}$.
Therefore, the QPIR rate is 
\begin{align}
R(\qprot{\vL}) =
\begin{cases}
\displaystyle
\frac{2\ell}{\vN \ell}  = \frac{2}{\vN}  & \text{if $\vN$ is even}\\
\displaystyle
\frac{2\ell}{(\vN+1)\ell}  = \frac{2}{\vN+1} &\text{if $\vN$ is odd}. 
% 2\ell/\vN \ell  = 2/\vN  & \text{if $\vN$ is even}\\
% 2\ell/(\vN+1)\ell  = 2/(\vN+1) &\text{if $\vN$ is odd}. 
\end{cases}
\end{align}
Moreover, 
the sequence $\{\Psi_{\mathrm{QPIR}}^{(\vL_\ell)}\}_{\ell=1}^{\infty}$ of our protocols for $\vL_{\ell} \coloneqq 2^{2\ell}$
achieves the negligible upload cost with respect to the download cost, i.e., 
$$\lim_{\ell\to\infty} \frac{\vN\vF}{\vN\ell} = \lim_{\ell\to\infty} \frac{\vN\vF}{(\vN+1)\ell} = 0.$$

\section{Converse}  \label{sec:converse}
In this section, we prove the converse bound \eqref{H2}
\begin{align}
C_{\mathrm{asymp}}^{\alpha,0,0,\theta} & \le \frac{2}{\vN} 
\label{ineq:converse_im}
\end{align}
for any $\alpha\in [0, 1)$, $\theta\geq 0$.
%Our converse proof gives the strict upper bound for any 

The idea of the converse proof is described as follows.
The converse bound \eqref{H2} is for the case where the QPIR protocol satisfies the user and server secrecy conditions perfectly, i.e., $\beta(\qprot{\vL})=0$ and $\gamma(\qprot{\vL})=0$.
From these perfect secrecies, we prove a lemma that the joint state of colluding servers is independent of the queried file $W_K$.
Then, using the lemma, the state from colluding servers can be used as shared entanglement between the honest server and the user, i.e., the honest server can communicate at most $2c$ bits to the user by sending $c$ qubits, which follows from the entanglement-assisted classical capacity \cite{BSST99}.
Since the user downloads $c \vN$ qubits from $\vN$ servers, the QPIR rate is upper bounded by $2/\vN$, which implies the converse bound \eqref{H2}.
Based on this idea, we give the converse proof in the remainder of this section.

% When $\cH'$ denotes the composite system of any $\vN-1$ systems in $\cH_1,\ldots,\cH_{\vN}$,
% Eq. \eqref{eq:serv_sec} implies
% \begin{align}
% I(W_i;\cH'|K=k)=0 \quad \forall i,k \in\{1,\ldots, \vF\} \text{\enskip s.t. }i\neq k. \label{eq:serv_sec2}
% \end{align}

% Let $s$ be any element of $\{1,\ldots,\vN\}$.
First, we prepare the following lemma.
%we introduce the following notations.
Recall that $Q_t'$ is the collection of queries to all servers other than $\mathtt{serv}_t$ for any $t\in\{1,\ldots,\vN\}$.
We denote by $\cH_t'$ the composite system of all servers other than $\mathtt{serv}_t$
% and the state on $\cH'$.
and by $d_t$ the dimensions of $\cH_t$.

\begin{lemm}\label{L1}
Suppose that $\beta(\qprot{\vL})=0$ and $\gamma(\qprot{\vL})=0$.
Let $\rho_{t|k}'$ be the state on $\cH_t'$ after the server encoder.
Then, the relation $I(W_k;\cH_t')_{\rho_{t|k}'}=0$ holds for any $k \in \{1,\ldots, \vF\}$ 
after the application of the server encoder.
That is, the state on the system $\cH_t'$ does not depend on the file information $W_k$.
%In particular, when $\gamma=0$, we have
%Then, the relation $I(W_k;\cH_s')_{\rho_{s|k}'}\le \beta$ holds for any $k \in \{1,\ldots, \vF\}$. 
\end{lemm}
\begin{proof}
%Let $\check{\mathsf{Enc}}_{\mathrm{serv}_s}$ be the collection $(\mathsf{Enc}_{\mathrm{serv}_{j}})_{j\neq s}$ of encoders except for $\mathsf{Enc}_{\mathrm{serv}_{s}}$.
Due to the condition \eqref{eq:user_sec},
the uploaded information $Q_t'$ is independent of $K$.
Since the $\rho_{t|k}'$ is determined by $Q_t'$, 
we have $\rho_{t|k}'=\rho_{t|i}'$, which implies that 
$I(W_k;\cH_t')_{\rho_{t|k}'}=I(W_k;\cH_t')_{\rho_{t|i}'}$ for $i\neq k \in\{1,\ldots, \vF\}$.
Since server secrecy (Eq.\eqref{eq:serv_sec}) implies
\begin{align}
I(W_k;\cH_t')_{\rho_{t|i}'} =0 \quad \forall i\neq k \in\{1,\ldots, \vF\}  \label{eq:serv_sec2},
\end{align}
we have $I(W_k;\cH_t')_{\rho_{t|k}'}=0 $ for any $k \in \{1,\ldots, \vF\}$. 
\end{proof}

We also prepare the following propositions.
\begin{prop} 
	[{\cite[(4.66)]{Hay17}}]
	\label{prop:samm}
Consider $x \in \{0,\ldots, m-1\}$ is encoded to a quantum state $\rho_x$ on $\cH$
	and decoded by a POVM
%$\{\rho_x \mid x \in \{0,\ldots, m-1\} \}$ be a set of states on $\cH$ 
	$\mathsf{Y} = \{ \mathsf{Y}_x \mid x \in \{0,\ldots, m\}\} $ on $\cH$, where the measurement outcome $m$ denotes decoding failure.
%	$\mathsf{Y} = \{ \mathsf{Y}_x \mid x \in \{0,\ldots, m-1\}\} $ be a POVM on $\cH$.
Define the error probability by $\varepsilon \coloneqq 1 - (1/m) \sum_{x=0}^{m-1} \Tr \rho_x \mathsf{Y}_x$.
Then, for any $\sigma$ on $\cH$ and any $r\in(0,1)$, we have 
\begin{align}
(1-\varepsilon)^{1+r} m^r \le \frac{1}{m}   \sum_{x=0}^{m-1} \Tr \rho_x^{1+r} \sigma^{-r}.
\end{align}
\end{prop}

%We also prepare the following data processing inequality \cite[(5.53)]{Hay17}.
\begin{prop}[{\cite[(5.53)]{Hay17}}] \label{ineq:infoprr}
%Let $\phi(-r|\rho\|\sigma) \coloneqq \log \Tr \rho^{1+r}\sigma^{-r}$.
For any states $\rho$ and $\sigma$, TP-CP map $\kappa$ and $r\in(0,1)$, we have 
\begin{align*}
\phi(-r|\rho\|\sigma) \coloneqq \log \Tr \rho^{1+r}\sigma^{-r} \geq \phi(-r|\kappa(\rho)\|\kappa(\sigma)).
\end{align*}
\end{prop}

\begin{prop} \label{prop:dxpuresmax}
For a pure state $|\psi\rangle$ in a system $\cH_1\otimes \cH_2$,
	let $\rho_1$ be the reduced state of $|\psi\rangle$ on $\cH_1$
	and $d_{\mathrm{min}} = \min \{ \dim\cH_1, \dim\cH_2\}$.
For any $s\in(0,1)$, we have 
	\begin{align}
	\Tr \rho_1^{s} \le d_{\mathrm{min}}^{1-s}.
	\label{eq:rho1s}
	\end{align}
\end{prop}
The proof of Proposition~\ref{prop:dxpuresmax} is given in Appendix~B.

%The proofs of Propositions~\ref{prop:samm} and \ref{prop:dxpuresmax} are given in the supplemental material.

Now, we prove the converse bound by three steps with the assumption that $\beta(\qprot{\vL})=0$ and $\gamma(\qprot{\vL})=0$.

\noindent{\textbf{Step 1}:}\quad
In this step, we introduce several notations for simplicity.
We denote 
	the queried file,
	the collection of non-queried $\vF-1$ files,
	the collection of queries
	by $w$, $w^c$, $q$, respectively,
	and we also denote $z \coloneqq (w^c,q)$.
Then, $\rho_{w,z}$ denotes the quantum state on $\cH = \bigotimes_{i=1}^{\vN} \cH_i$.
We consider the entire quantum systems $\cH$ as a bipartite system $\cH = \cH_t\otimes \cH_t'$.
We also denote by $\rho_{z}'$ the reduced state of $\rho_{w,z}$ on $\cH_t'$
	since the reduced state does not depend on $w$ from Lemma~\ref{L1}.

\noindent{\textbf{Step 2}:}\quad
In this step, we prove the inequality
\begin{align}
(1-P_{\mathrm{err}}(\Psi_{\mathrm{QPIR}}^{(\vL)}) )^{1+r} \vL^{r}
 \le \min_t d_t^{2r}.
 \label{asdf1}
\end{align}

Applying Proposition~\ref{prop:samm} to the choice 
	\begin{align}
	(x,\rho_x, \sigma) &\coloneqq \paren*{w , \rho_{w,z}, \rho_{z}' \otimes \frac{\mathsf{I}_{\cH_t}}{d_t}}
	\end{align}
	for any $r \in (0,1)$, we obtain
	\begin{align}
	(1-P_{\mathrm{err},z}(\Psi_{\mathrm{QPIR}}^{(\vL)}))^{1+r} \vM^r \le \frac{1}{\vM}   \sum_{w=0}^{\vM-1} \Tr \rho_{w,z}^{1+r} \paren*{\rho_{z}' \otimes \frac{\mathsf{I}_{\cH_t}}{d_t}}^{-r}
	\end{align}
where $P_{\mathrm{err},z}(\Psi_{\mathrm{QPIR}}^{(\vL)})$ is the error probability when $z$ is fixed.
By averaging with respect to $z$ and from the convexity of $x^{1+r}$,	
	we have
\begin{align}
&(1-P_{\mathrm{err}}(\Psi_{\mathrm{QPIR}}^{(\vL)}) )^{1+r} \vL^{r}
%\le \mathbb{E}_{Z} \frac{1}{\vL} \sum_{w=0}^{\vL-1} \Tr \rho_{w,Z}^{1+r} {\sigma}_Z^{-r} \nonumber \\
 \le \mathbb{E}_{Z} \frac{1}{\vL} \sum_{w=0}^{\vL-1} \Tr \rho_{w,Z}^{1+r} \paren*{\rho_{Z}' \otimes \frac{\mathsf{I}_{\cH_t}}{d_t}}^{-r}
\label{H4}
\end{align}
for any $t\in\{1,\ldots,\vN\}$.

Let $t^* = \argmin_t d_t$.
We choose orthonormal vectors $\{|\psi_{w,z,x}\rangle \}_x$ such that $\rho_{w,z}=\sum_x p_{w,z,x} |\psi_{w,z,x}\rangle \langle\psi_{w,z,x}|$ with $\sum_x p_{w,z,x} = 1$.
%For diagonalization $\rho_{w,z}=\sum_x p_{w,z,x} |\psi_{w,z,x}\rangle \langle\psi_{w,z,x}|$ by orthonormal vectors $|\psi_{w,z,x}\rangle$,
We denote by $\sigma_{z,x}'$ the reduced state of $|\psi_{w,z,x}\rangle \langle\psi_{w,z,x}|$ on $ \cH_{t^*}'$, {i.e., $\rho_z' = \sum_x p_{w,z,x} \sigma_{z,x}'$.}
With this decomposition, we can upper bound the RHS of \eqref{H4} as 
\begin{align}
& \Tr \rho_{w,z}^{1+r} \paren*{\rho_{z}' \otimes \frac{\mathsf{I}_{\cH_{t^*}}}{d_{t^*}}}^{-r} \nonumber \\
&\le  \sum_x p_{w,z,x}
\Tr (|\psi_{w,z,x}\rangle \langle\psi_{w,z,x}|)^{1+r} \paren*{\sigma_{z,x}' \otimes \frac{\mathsf{I}_{\cH_{t^*}}}{d_{t^*}}}^{-r} \!\!\! \label{ineq:rev} \\ 
&= \sum_x p_{w,z,x} 
\Tr |\psi_{w,z,x}\rangle \langle\psi_{w,z,x}| \paren*{\sigma_{z,x}' \otimes \frac{\mathsf{I}_{\cH_{t^*}}}{d_{t^*}}}^{-r} \nonumber \\
&= d_{t^*}^r 
	\sum_x 
p_{w,z,x} 
\Tr |\psi_{w,z,x}\rangle \langle\psi_{w,z,x}| ((\sigma_{z,x}')^{-r} \otimes \mathsf{I}_{\cH_{t^*}})
\nonumber \\
&= d_{t^*}^r 
	\sum_x 
p_{w,z,x} 
	(\Tr_{\cH_{t^*}} |\psi_{w,z,x}\rangle \langle\psi_{w,z,x}|) (\sigma_{z,x}')^{-r}
	\nonumber \\
&= d_{t^*}^r \sum_x p_{w,z,x} 
\Tr (\sigma_{z,x}')^{1-r} 
\label{fromComment} \\
&\le d_{t^*}^{2r}. \label{H5}
\end{align}
Here, 
{Eq.~\eqref{ineq:rev} is obtained by applying Proposition~\ref{ineq:infoprr} for the choice 
	\begin{align*}
	\rho &\coloneqq \sum_x p_{w,z,x} |x\rangle\langle x|\otimes |\psi_{w,z,x}\rangle\langle \psi_{w,z,x}|,	\\
	\sigma &\coloneqq \sum_x p_{w,z,x} |x\rangle\langle x|\otimes \paren*{\sigma_{z,x}\otimes \frac{\mathsf{I}_{\cH_{t^*}}}{d_{t^*}}}, \\
		%on $\mathcal{X}\otimes\cH$,
	\kappa &\coloneqq\Tr_{\mathcal{X}},
	\end{align*}
	since 
	$\phi(-r|\rho\|\sigma)$ is the RHS of \eqref{ineq:rev} and 
	$\phi(-r|\kappa(\rho)\|\kappa(\sigma))$ is the LHS of \eqref{ineq:rev} from
	\begin{align*}
		\kappa(\rho) &= \Tr_{\mathcal{X}} \sum_x p_{w,z,x} |x\rangle\langle x|\otimes |\psi_{w,z,x}\rangle\langle \psi_{w,z,x}| \\
		&= 
		\sum_x p_{w,z,x} |\psi_{w,z,x}\rangle\langle \psi_{w,z,x}| = \rho_{w,z},\\ 
		\kappa(\sigma) &= \Tr_{\mathcal{X}} \sum_x p_{w,z,x} |x\rangle\langle x|\otimes \paren*{\sigma_{z,x}\otimes \frac{\mathsf{I}_{\cH_{t^*}}}{d_{t^*}} } \\
		&=
		\sum_x p_{w,z,x} \sigma_{z,x}\otimes \frac{\mathsf{I}_{\cH_{t^*}}}{d_{t^*}} = \rho_z' \otimes \frac{\mathsf{I}_{\cH_{t^*}}}{d_{t^*}}.
	\end{align*}
	%and
	%\begin{align}
	%\phi(-r|\rho\|\sigma) 
	%		=	\Tr \rho^{1+r}\sigma^{-r}  
	%		=	\Tr
	%\end{align}
Eq.~\eqref{fromComment} follows from $\Tr |\psi_{w,z,x}\rangle \langle\psi_{w,z,x}| = \sigma_{z,x}'$.
Eq.~\eqref{H5} is obtained 
	by applying Proposition~\ref{prop:dxpuresmax}
		for $(|\psi\rangle, \rho_1, s, d_{\mathrm{min}}) \coloneqq  (|\psi_{w,z,x}\rangle, \sigma_{z,x}', 1-r, \dim \cH_{t^*})$.
%	because 
%% 	a pure state on a bipartite system 
%	the state $\sigma_{z,x}'$ is the reduced state of a pure state $|\psi_{w,z,x}\rangle$ in $\cH_t'\otimes \cH_t$,
%	which implies $\rank \sigma_{z,x}' \le \dim \cH_t = d_t$ and therefore 
%% Since the rank of $\sigma_{z,x}$ is not greater than $d_t$, we have
Combining \eqref{H4} for $t \coloneqq t^*$ and \eqref{H5}, we have the desired inequality \eqref{asdf1}.

\noindent{\textbf{Step 3}:}\quad
In this step, we prove the converse bound by contradiction.
Suppose that there exists a sequence of QPIR protocols $\{\Psi_{\mathrm{QPIR}}^{(\vL_\ell)}\}_{\ell=1}^{\infty}$ 
% \substack{\{\vL_\ell\}_{\ell=1}^{\infty},\\ 
	such that 
		%the error and  
		%the QPIR rate satisfy 
		\begin{align}
		\limsup_{\ell\to\infty} P_{\mathrm{err}}(\Psi_{\mathrm{QPIR}}^{(\vL_\ell)}) &< 1,  \label{eq:sdperr}\\
		\liminf_{\ell\to\infty} R(\Psi_{\mathrm{QPIR}}^{(\vL_\ell)}) &> \frac{2}{\vN}.
		\end{align}
This inequality is also written as 
\begin{align*}
	&1 > \limsup_{\ell\to\infty} \frac{2}{\vN R(\Psi_{\mathrm{QPIR}}^{(\vL_\ell)})}
		=
	  \limsup_{\ell\to\infty} \frac{2}{\vN}\cdot \frac{\log D(\Psi_{\mathrm{QPIR}}^{(\vL_\ell)})}{\log \vL_\ell} \\
	  &= 
	  \limsup_{\ell\to\infty} \frac{\log (D(\Psi_{\mathrm{QPIR}}^{(\vL_\ell)}))^{2/\vN}}{\log \vL_\ell},
\end{align*}
which implies 
$$\limsup_{\ell\to\infty} \frac{(D(\Psi_{\mathrm{QPIR}}^{(\vL_\ell)}))^{2/\vN}}{\vL_\ell} = 0.$$
Furthermore, 
	since 
	$$
	0
	\leq 
	% ((\min_t \dim \cH_t)^{\vN})^{2/\vN}/\vL_\ell= 
	\frac{\min_t d_t^2}{\vL_\ell}
	=
	\frac{(\min_t d_t^\vN)^{2/\vN}}{\vL_\ell}
	\le
	\frac{(D(\Psi_{\mathrm{QPIR}}^{(\vL_\ell)}))^{2/\vN}}{\vL_\ell}
	,$$
% (D(\Psi_{\mathrm{QPIR}}^{(\vL_\ell)}))^{2/\vN}/\vL_\ell 
% \ge (1-P_{\mathrm{err}}(\Psi_{\mathrm{QPIR}}^{(\vL_\ell)}) )^{1+1/r} 
	we have
\begin{align}
\lim_{\ell\to\infty} \frac{\min_t d_t^2}{\vL_\ell} = 0.
	\label{asdf2}
\end{align}
Combining \eqref{asdf1} and \eqref{asdf2}, the probability of correctness $1-P_{\mathrm{err}}(\Psi_{\mathrm{QPIR}}^{(\vL_\ell)})$ approaches $0$, which contradicts with \eqref{eq:sdperr}.
Thus, any QPIR protocol with asymptotic error probability less than $1$ has QPIR rate at most $2/\vN$,
	which implies \eqref{H2}.

% By using the independence of $\rho'$ from $W_K$, the converse is proved as follows.
% After $\cH'$ is transmitted to the user,
% we can consider the downloading step of our protocol 
% as an instance of entanglement-assisted classical communication of $W_K$ from $\mathtt{serv}_{s}$ to the user in the case that all servers except for $\mathtt{serv}_s$ for any $s \in \{1,\ldots, \vL\}$ collude.
% Therefore, 
% by the entanglement-assisted classical capacity \cite{BSST99} for the identity map $\mathsf{I}_{\cH_{s}}$, we have
% \begin{align}
% \lefteqn{\log \vL \leq C_{\mathrm{EA}}(\mathsf{I}_{\cH_{s}}) = \max_{|x_{\cH_{s}\!\mathcal{R}}\rangle} I_{|x_{\cH_{s}\!\mathcal{R}}\rangle} (\cH_{s};\mathcal{R})}\\
% =& \max_{|x_{\cH_{s}\!\mathcal{R}}\rangle} H_{|x_{\cH_{s}\!\mathcal{R}}\rangle}(\cH_s)+H_{|x_{\cH_{s}\!\mathcal{R}}\rangle}(\mathcal{R}) \leq 2\log \dim \cH_{s},
% \end{align}
% where $\mathcal{R}$ is a reference system and the maximization is with respect to all bipartite pure states $|x_{\cH_{s}\!\mathcal{R}}\rangle$ on $\cH_s\otimes \mathcal{R}$.
% Summing this inequality for any $s\in\{1,\ldots,\vN\}$, we have
% \begin{align}
% \vN \log \vL \leq 2 \log \dim \bigotimes_{t=1}^{\vN} \cH_i,
% \end{align}
% which proves the converse part of Theorem \ref{theo:main}.

\section{Conclusion}

We have presented the $(\vN-1)$-private QSPIR capacity for even number $\vN$ of servers when any $\vN-1$ servers collude.
We constructed a $(\vN-1)$-private QSPIR protocol of rate $\ceil{\vN/2}^{-1}$, and proved that $\vN/2$ is the strong converse bound with the perfect server and user secrecy.
Our protocol is constructed by using the quantum teleportation and the two-sum transmission protocol repetitively.
% The converse bound used the fact that the state of any $\vN-1$ servers is independent of the retrieved file from the perfect secrecy of the server and user.
The converse bound used the fact that the state of any $\vN-1$ servers is independent of the retrieved file, which follows from the perfect secrecy of the server and the user.
% and therefore, the state of any $\vN-1$ servers can be considered as a 

	Following the conference version of this paper, the paper \cite{SH20} proved that both symmetric and non-symmetric $\vT$-private QPIR capacity $1$ for $\vT\leq \vN/2$.
		%, the non-symmetric $\vT$-private QPIR capacity equals to .
	However, for $\vT > \vN/2$, the $\vT$-private QPIR capacity is derived only for the case where the server secrecy, i.e., leakage to the user, is negligible with respect to the size of files.
	Thus, it is an open problem to derive the (non-symmetric) $\vT$-private QPIR capacity for $\vT > \vN/2$.
	For this problem, the converse proof of classical $\vT$-private PIR \cite{SJ18} cannot be directly applied 
		%to the quantum case 
		because the $\vT$-private QPIR capacity for $\vT\leq \vN/2$ is $1$ and it is already greater than the classical capacity $\paren*{1-\vT/\vN}/\big(1-\paren*{\vT/\vN}^{\vF}\big)$ of \cite{SJ18}.
	Thus, we expect that the $\vT$-private QPIR capacity for $\vT > \vN/2$ is also greater than its classical capacity of \cite{SJ18}.
	We leave this problem as an open problem.

\section*{Acknowledgement}
The authors are grateful to Dr. Hsuan-Yin Lin and Dr. Eirik Rosnes for helpful discussions and comments for the comparison of classical PIRs to quantum PIRs \cite{HE19}.

\appendices

\section{Fundamentals of Quantum Information Theory} \label{append:basic}

In this appendix, we briefly introduce basic concepts of the quantum information theory.
In the following, we give the mathematical definitions of quantum system, quantum state, quantum measurement, and quantum operation. 
For the physical motivations of these definitions and more detailed introduction of these definitions, we refer to \cite{NC00,Hay17}.

A quantum system, or simply {\em system}, is an object considered in quantum information theory
	and it is mathematically defined by a complex Hilbert space $\cH$. 
In this paper, we only consider finite-dimensional Hilbert spaces.
A two-dimensional Hilbert space is called a qubit.

The state of a quantum system, also called quantum state, represents the the information in a quantum system.
A quantum state of $\cH$ is mathematically described by 
	a density matrix $\rho$ defined by a positive-semidefinite matrix on $\cH$ such that $\Tr\rho = 1$.
As a special case, the state $\rho = (1/\dim \cH)\cdot  \mathsf{I}_{\cH}$ is called the {\em completely mixed state}.
When a density matrix $\rho$ is a rank-one matrix, we can denote $\rho = |\psi\rangle\langle\psi|$ with a unit vector $|\psi\rangle$ and thus, we sometimes represent the state by the unit vector $|\psi\rangle$, called a pure state.
A quantum state is mapped to another state by quantum operations.

For two quantum systems $\cH_1$ and $\cH_2$, the composite quantum system is described by $\cH_1\otimes \cH_2$.
When the state of $\cH_1\otimes \cH_2$ is $\rho$, the reduced state on $\cH_1$ is represented by $\rho_1\coloneqq \Tr_{\cH_1} \rho$, where $\Tr_{\cH_1}$ is the partial trace on $\cH_1$ defined by the unique linear map such that $\Tr_{\cH_1} X\otimes Y = (\Tr X) Y$ for any matrices $X$ on $\cH_1$ and $Y$ on $\cH_2$.
%Multiple quantum systems are described by the tensor product of each Hilbert space.
%On two qubits $\cH_1$ and $\cH_2$, 
The Schmidt decomposition is the decomposition theorem for any pure state $|\psi\rangle$ on a composite system $\cH_1\otimes \cH_2$ as follows.
	\begin{prop}[Schmidt decomposition] \label{prop:schmidt}
	For any pure state $|\psi\rangle$ on a composite system $\cH_1\otimes \cH_2$,
	there exist orthonormal vectors $\{|x_i\rangle \}_i$ of $\cH_1$ and $\{|y_i\rangle \}_i$ of $\cH_2$ 
		and a probability distribution $\{p_i\}_i$
		such that 
	$|\psi\rangle = \sum_i \sqrt{p_{i}} |x_i\rangle \otimes |y_i\rangle \in \cH_1 \otimes \cH_2$.
	\end{prop}
%Thus, the reduced states $\rho_1 = \sum_{i} p_i |x_i\rangle\langle x_i|$ on $\cH_1$ and $\rho_2 = \sum_{i} p_i |y_i\rangle\langle y_i|$  on $\cH_2$ have the same rank
%and therefore, $\rank \rho_1 \leq  \min \{ \dim\cH_1, \dim\cH_2\} = d_{\mathrm{min}}$.

A quantum operation is described by 
	a trace-preserving and completely positive (TP-CP) map.
A map $\kappa$ from a system $\cH$ to a system $\cH'$ is called a TP-CP map if
	$\kappa$ is linear,
	$\Tr\kappa(\rho) = \rho$ for any state $\rho$ on $\cH$,
	and
	$\kappa\otimes \iota_{\mathbb{C}^{d}} (\rho)$ is a positive-semidefinite matrix for any state $\rho$ on $\cH\otimes \mathbb{C}^{d}$, 
	where $\iota_{\mathbb{C}^d}$ is the identity map on the system $\mathbb{C}^d$.
The simplest example of a TP-CP map is a unitary map $\rho\mapsto U\rho U^\dagger$ for a unitary matrix $U$.
When a unitary map of a unitary matrix $U$ is applied to the system with pure state $|\psi\rangle$, 
	the resultant state is the pure state $U|\psi\rangle$.
	
Classical information is extracted from a system $\cH$ by a quantum measurement.
A quantum measurement on a system $\cH$ is described by a positive-operator valued measure (POVM).
A POVM on $\cH$ is a set $\mathsf{Y} = \{\mathsf{Y}_i \mid i\in \mathcal{I}\}$ of positive-semidefinite matrices on $\cH$ such that $\mathsf{Y}_i = \mathsf{I}_{\cH}$.
When the POVM measurement $\mathsf{Y}$ is performed on the system $\cH$ with state $\rho$,  
	the measurement outcome $i\in\mathcal{I}$ is obtained with probability $\Tr \rho \mathsf{Y}_i \in [0,1]$.
If all elements of POVM are projections, i.e., $\mathsf{Y}_i^2 = \mathsf{Y}_i$ and $\mathsf{Y}_i^{\dagger} = \mathsf{Y}_i$,
	the POVM is also called a projection-valued measure (PVM).
Sometimes an orthonormal basis $\{ |e_i\rangle \}_i$ of $\cH$ is considered as a PVM $\{ |e_i\rangle\langle e_i| \}_i$.

\section{Proof of Proposition~\ref{prop:dxpuresmax}}
%\begin{proof}
%Let $\rho_2$ be the reduced state of $|\psi\rangle$ on $\cH_2$.
By the Schmidt decomposition (Proposition~\ref{prop:schmidt}),
	the pure state $|\psi\rangle$ on $\cH_1\otimes \cH_2$
	is decomposed as
	$|\psi\rangle = \sum_i \sqrt{p_{i}} |x_i\rangle \otimes |y_i\rangle \in \cH_1 \otimes \cH_2$
	for some orthonormal vectors $\{|x_i\rangle \}_i$ of $\cH_1$ and $\{|y_i\rangle \}_i$ of $\cH_2$ 
		and some probability distribution $\{p_i\}_i$.
Thus, the reduced states $\rho_1 = \sum_{i} p_i |x_i\rangle\langle x_i|$ on $\cH_1$ and $\rho_2 = \sum_{i} p_i |y_i\rangle\langle y_i|$  on $\cH_2$ have the same rank
and therefore, $\rank \rho_1 \leq  \min \{ \dim\cH_1, \dim\cH_2\} = d_{\mathrm{min}}$.
On the other hand, 
	since the function $x^s$ for $s\in(0,1)$ is concave,
		the LHS of \eqref{eq:rho1s} $\Tr \rho_1^{s} = \sum_i p_i^s$ maximizes if $p_i$ have the same value for all $i$, i.e., $\{p_i\}_i$ is the uniform distribution, by Karamata's inequality \cite{Karamata}. 
Thus, we have the desired inequality as 
	%\begin{equation*}
	$$\Tr \rho_1^{s} = \sum_i p_i^s  \leq \sum_{j=1}^{d_{\mathrm{min}}} d_{\mathrm{min}}^{-s} = d_{\mathrm{min}}^{1-s}. 
	$$
%	\QEDB
	%{\flushright  \QEDB}
	%\end{equation*}
%\end{proof}

\begin{IEEEbiographynophoto}{Seunghoan Song}(GS'20)
received the B.E. degree from Osaka University, Japan, in 2017 and the M.Math.\ and Ph.D.\ degrees in mathematical science from Nagoya University, Japan, in 2019 and 2020, respectively. 
He is a research fellow of the Japan Society of the Promotion of Science (JSPS) from 2020.
He is currently a JSPS postdoctoral fellow at the Graduate School of Mathematics, Nagoya University.
He awarded the School of Engineering Science Outstanding Student Award in 2017 and 
Graduate School of Mathematics Award for Outstanding Masters Thesis in 2019.
His research interests include classical and quantum information theory and its applications to secure communication protocols.
\end{IEEEbiographynophoto}

\begin{IEEEbiographynophoto}{Masahito Hayashi}(M'06--SM'13--F'17) was born in Japan in 1971.
He received the B.S.\ degree from the Faculty of Sciences in Kyoto
University, Japan, in 1994 and the M.S.\ and Ph.D.\ degrees in Mathematics from
Kyoto University, Japan, in 1996 and 1999, respectively. He worked in Kyoto University as a Research Fellow of the Japan Society of the
Promotion of Science (JSPS) from 1998 to 2000,
and worked in the Laboratory for Mathematical Neuroscience,
Brain Science Institute, RIKEN from 2000 to 2003,
and worked in ERATO Quantum Computation and Information Project,
Japan Science and Technology Agency (JST) as the Research Head from 2000 to 2006.
He also worked in the Superrobust Computation Project Information Science and Technology Strategic Core (21st Century COE by MEXT) Graduate School of Information Science and Technology, The University of Tokyo as Adjunct Associate Professor from 2004 to 2007.
He worked in the Graduate School of Information Sciences, Tohoku University as Associate Professor from 2007 to 2012.
In 2012, he joined the Graduate School of Mathematics, Nagoya University as Professor.
Also, he was appointed in Centre for Quantum Technologies, National University of Singapore as Visiting Research Associate Professor from 2009 to 2012
and as Visiting Research Professor from 2012 to now.
He worked in Center for Advanced Intelligence Project, RIKEN as
a Visiting Scientist from 2017 to 2020.
He worked in Shenzhen Institute for Quantum Science and Engineering, Southern University of Science and Technology, Shenzhen, China as a Visiting Professor from 2018 to 2020,
and
in Center for Quantum Computing, Peng Cheng Laboratory, Shenzhen, China
as a Visiting Professor from 2019 to 2020.
In 2020, he joined Shenzhen Institute for Quantum Science and Engineering, Southern University of Science and Technology, Shenzhen, China
as Chief Research Scientist. 
In 2011, he received Information Theory Society Paper Award (2011) for ``Information-Spectrum Approach to Second-Order Coding Rate in Channel Coding''.
In 2016, he received the Japan Academy Medal from the Japan Academy
and the JSPS Prize from Japan Society for the Promotion of Science.

In 2006, he published the book ``Quantum Information: An Introduction'' from Springer, whose revised version was published as ``Quantum Information Theory: Mathematical Foundation'' from Graduate Texts in Physics, Springer in 2016.
In 2016, he published other two books ``Group Representation for Quantum Theory'' and ``A Group Theoretic Approach to Quantum Information'' from Springer.
He is on the Editorial Board of {\it International Journal of Quantum Information}
and {\it International Journal On Advances in Security}.
His research interests include classical and quantum information theory and classical and quantum statistical inference.
\end{IEEEbiographynophoto}

\end{document}